\documentclass[submission,copyright,creativecommons]{eptcs}
\providecommand{\event}{GandALF 2021} 
\usepackage{breakurl}             
\usepackage{underscore}           

\usepackage{listings}
\lstdefinelanguage{pseudo}{morekeywords={init,with,or,if,then,else,fi,and,not,while,do,od,distinct,
    case, goto,local,algorithm, function, for, each, times, from, to,
    variables, procedure, recursive, return},
  morecomment=[l]{//}, morecomment=[s]{/*}{*/},
  mathescape=true,escapechar={@},
  basicstyle=\sffamily\small,
  commentstyle=\itshape\rmfamily\small,
  keywordstyle=\sffamily\bfseries\small
}
\usepackage{mathtools}
\usepackage{amsfonts,amssymb,amsmath, amsthm}

\usepackage{etoolbox}  
\usepackage{enumitem} 

\usepackage{stmaryrd}
\usepackage{bm}

\usepackage{hyperref}
\usepackage{environ}
\usepackage{csvsimple}
\usepackage{longtable} 
\usepackage[normalem]{ulem} 

\usepackage{fourier-orns} 
\usepackage{wasysym} 
\usepackage{marvosym} 
\usepackage{bbding} 

\usepackage{etoolbox}  
\usepackage{enumitem} 

\usepackage{enumerate}
\usepackage{mathrsfs}
\usepackage{mathtools}
\usepackage{appendix}
\usepackage{comment}
\usepackage{tikz}
\usepackage{array}
\usetikzlibrary{calc}
\usepackage{multicol}

\usetikzlibrary{arrows,shapes,snakes,automata,backgrounds,petri,positioning}
\usetikzlibrary{fit,backgrounds} 
\definecolor{processblue}{cmyk}{0.96,0,0,0}

\newcommand{\be}{\begin{enumerate}}
\newcommand{\ee}{\end{enumerate}}
\newcommand{\bc}{\begin{center}}
\newcommand{\ec}{\end{center}}

\newcommand{\bi}{\begin{itemize}}
\newcommand{\ei}{\end{itemize}}

\newcommand{\act}{\xrightarrow}

\newtheorem{theorem}{Theorem}

\newtheorem{prop}[theorem]{Proposition}

\newtheorem{definition}[theorem]{Definition}
\newtheorem{example}{Example}

\usepackage{hyperref}

\newcommand\slice[2]{#1{\raise-.5ex\hbox{\ensuremath|}}_{#2}}

\newcommand{\defeq}{\stackrel{\scriptscriptstyle\text{def}}{=}}

\newcommand{\N}{\mathbb{N}}                    
\renewcommand{\vec}[1]{\bm{#1}}                
\newcommand{\set}[1]{\left\{#1\right\}}        

\newcommand{\multiset}[1]{\Lbag#1\Rbag}        
\renewcommand{\vec}[1]{\bm{#1}}                
\newcommand{\norm}[1]{\lVert#1\rVert}          

\newcommand{\support}[1]{\norm{#1}} 

\usepackage{thm-restate}

\newcommand{\net}{\mathcal{N}}

\newcommand{\trans}[1]{\xrightarrow{#1}}       
\newcommand{\pre}{\mathit{pre}} 
\newcommand{\post}{\mathit{post}} 
\newcommand{\prestar}{\mathit{pre}^*}
\newcommand{\poststar}{\mathit{post}^*}

\renewcommand{\norm}[1]{\| {#1} \|}

\newcommand{\unorm}[1]{\|{#1}\|_u}
\newcommand{\lnorm}[1]{\|{#1}\|_l}
\newcommand{\sem}[1]{\llbracket{#1}\rrbracket} 

\newcommand{\cube}{\mathcal{C}}
\newcommand{\cC}{\Gamma}      
\newcommand{\cSet}{\mathcal{S}}  

\newcommand{\RBN}{\mathcal{R}}

\newcommand{\ASMS}{\mathcal{P}}
\newcommand{\oper}{\mathtt{op}}
\newcommand{\data}{\textit{data}}

\makeatletter

\newcommand{\eqxrightarrow}[2]{%
  \mathop{%
    \vtop{%
      \m@th 
      \offinterlineskip 
      \ialign{%
        \hfil##\hfil\cr
        \rightarrowfill\cr
        \hphantom{$\scriptstyle\mskip8mu{#2}\mskip8mu$}\cr
        \vrule height0pt width 1.5em\cr
        $\scriptscriptstyle {#1}$\cr
      }%
    }%
  }\limits^{#2}%
}
\makeatother

\usepackage[disable]{todonotes}
\newcommand{\chana}[1]{\todo[color=green!30]{\small #1}}

\title{Reconfigurable Broadcast Networks and Asynchronous Shared-Memory Systems are Equivalent}
\author{A. R. Balasubramanian \qquad\qquad Chana Weil-Kennedy
\institute{Technical University of Munich \\ Munich, Germany \thanks{This project has received funding from the European Research Council (ERC) under the European Union's Horizon 2020 research and innovation programme under grant agreement No 787367 (PaVeS).}}
\email{bala.ayikudi@tum.de \quad\qquad chana.weilkennedy@in.tum.de}
}
\def\titlerunning{RBN and ASMS are Equivalent}
\def\authorrunning{A. R. Balasubramanian, C. Weil-Kennedy}

\begin{document}
\maketitle

\begin{abstract}
We show the equivalence of two distributed computing models, namely
reconfigurable broadcast networks (RBN) and asynchronous shared-memory systems (ASMS),
that were introduced independently.
Both RBN and ASMS are systems in which a collection of anonymous, finite-state processes
run the same protocol. In RBN, the processes communicate by selective broadcast: 
a process can broadcast a message which is received by all of its neighbors, and the set
of neighbors of a process can change arbitrarily over time.
In ASMS, the processes communicate by shared memory: a process can either write to or read from 
a shared register.
Our main result is that RBN and ASMS can \emph{simulate} each other, i.e. they are equivalent with respect
to parameterized reachability, where we are given two (possibly infinite) sets of configurations $\cube$ and $\cube'$ 
defined by upper and lower bounds on the number of processes in each state and we would like to decide if some configuration in $\cube$ can reach
some configuration in $\cube'$.
Using this simulation equivalence, we transfer results of RBN to ASMS and vice versa.
Finally, we show that RBN and ASMS can simulate a third distributed model
called immediate observation (IO) nets. Moreover, for a slightly stronger notion of simulation (which is satisfied by all the simulations 
given in this paper), we show that IO nets cannot simulate RBN.
\end{abstract}

\section{Introduction}

In this paper, we consider three models of distributed computation, one in which
communication happens by (selective) broadcasts, another in which communication happens by means of a shared memory, and finally one in which
communication happens by observation. We first expand a bit more on these models, then describe our main results and finally derive some
consequences from these results.

The first model that we consider is \emph{reconfigurable broadcast networks} (RBN)\cite{AdHocNetworks, FSTTCS12}.
In this model, we have a collection of anonymous, finite-state processes executing the same protocol. Further, every process has a set of neighbors.
At each step, a process can broadcast a message which is then received by all of the processes in its neighborhood. 
The neighborhood topology is reconfigurable, meaning that the set of neighbors of a process can change arbitrarily between two steps.
Parameterized verification of RBN
aims to prove that a property is correct, irrespective
of the number of participating processes.
Dually, it attempts to find an execution of some population of processes
for which a property is violated.  Within this context, the complexity of different variants of (parameterized) reachability and repeated coverability have been studied for RBN \cite{AdHocNetworks, FSTTCS12, Liveness}. 
Moreover, many extensions of RBN with clocks, registers and probabilities have been proposed and studied, mainly within the perspective of parameterized verification \cite{register,prob,probtime}.

The second model that we consider is a formal model of \emph{asynchronous shared-memory systems} (ASMS)\cite{ModelCheckingSMS, ParamLiveness, ICALPPatricia}.
In this model, we have a collection of anonymous, finite-state processes executing the same protocol, and a single register which all processes can access 
to perform a read/write operation. The set of values that can be stored in this register is finite. 
No locks onto the register are allowed and so no process can perform a sequence of atomic operations whilst preventing other processes from accessing the register.
Similar to RBN, major questions of interest in ASMS are those pertaining to parameterized verification, i.e. finding bad executions over
some population of processes.
The complexity of some (parameterized) reachability and model-checking questions for ASMS have been explored in a series of papers
\cite{ModelCheckingSMS, JACM16, ParamLiveness}. Further extensions of this model with leaders, stacks, etc. have also been studied \cite{JACM16, FineGrained, SafetyAlmostAlways, ModelCheckingPushdown}. Finally, \cite{ICALPPatricia} considers ASMS in the setting in which a stochastic scheduler
picks a process (uniformly at random) at each step to be executed, and under this setting studies the question of 
whether a given state can be reached by some process almost-surely, i.e., with probability 1. 

The third model that we consider is \emph{immediate observation Petri nets} (IO nets) \cite{EsparzaGMW18,EsparzaRW19}, 
which were introduced with motivations from the field of population protocols \cite{Comp-Power-Pop-Prot,First-Pop-Prot}.
Intuitively, in this model, we have a collection of anonymous, finite-state processes executing the same protocol. The only communication allowed
between processes is \emph{observation}, i.e., a process can only observe that another process is at some other state, and based on this 
observation can execute a step. The process being observed cannot detect if some process is observing it. Motivated by application to population protocols,
the authors of \cite{EsparzaGMW18, EsparzaRW19} study parameterized reachability questions for IO nets.

In this paper, we show that RBN and ASMS can \emph{simulate} each other, with respect to (parameterized) reachability.
Roughly speaking, we show that any instance of a parameterized reachability question for RBN can be efficiently translated to an instance
of parameterized reachability for ASMS and vice versa. More specifically, we consider the question of \emph{cube-reachability}. 
In the cube-reachability question, we are given an instance of a model (which can be either an RBN, an ASMS or an IO net) and two sets of configurations $\cube,\cube'$,
each of them defined by lower and upper bounds on the number of processes in each state.  (The upper bounds on some states might be $\infty$, which means
that we allow arbitrary number of processes in that state). We would then like to decide if there is a configuration in $\cube$ which can 
reach a configuration in $\cube'$. As we shall explain in the next section, the cube-reachability question covers parameterized reachability and coverability problems,
parameterized reachability problems with leaders, and allows for a uniform transfer of results between the models that we study in this paper.

Our main result is that the cube-reachability questions for RBN and ASMS are polynomial-time equivalent to each other. This result, along with the constructions
achieving this result, enable us to translate various parameterized reachability results from RBN to ASMS and vice versa. First, we show that a special case of cube-reachability, called unbounded initial cube reachability,
is \textsf{PSPACE}-complete for ASMS, by using our reduction and already existing similar results on RBN. Then, we introduce the model of RBN-leader protocols and 
use already existing results on ASMS-leader protocols to prove that the RBN-leader reachability problem is \textsf{NP}-complete. Finally, we show that
the almost-sure coverability problem for RBN is decidable in \textsf{EXPSPACE} by translating the analogous result for ASMS~\cite{ICALPPatricia}.

Additionally, we show that the cube-reachability problem for IO nets reduces to the cube-reachability problem for RBN, leading to a transfer of results from RBN to IO nets. For the other direction, we actually show an impossibility result.
We define a stronger form of reduction for the cube-reachability problem and we notice that the 
reductions given in this paper all satisfy this stronger property. Then, using results from IO net theory, 
we show that there can be no reduction from the cube-reachability problem for RBN to the cube-reachability problem for IO nets which satisfies this stronger property.
We leave open the problem of whether there can exist other reductions from RBN to IO nets.

The rest of the paper is organized as follows: In Section 2, we present some preliminary definitions and notations, then in Section 3, we describe RBN and ASMS.
Section 4 proves our main result that RBN and ASMS can simulate each other. Section 5 presents some transfer of results between RBN and ASMS. 
In Section 6, we introduce IO nets, show that they can be simulated by RBN, and prove that the other direction is not true for a stronger form of simulation.
For space reasons, all missing proofs are relegated to the appendix.

\section{Preliminaries}


\paragraph*{Multisets.}
A \emph{multiset} on a finite set \(E\) is a mapping \(C \colon E \rightarrow \N\), i.e. for any $e\in E$, \(C(e)\) denotes the number of occurrences of element \(e\) in \(C\).
We let $\mathbb{M}(E)$ denote the set of all multisets on $E$.
Let $\multiset{e_1,\ldots,e_n}$ denote the multiset $C$ such that $C(e)=|\{j\mid e_j=e\}|$.
We sometimes write multisets using set-like notation. 
For example, $\multiset{2 \cdot a,b}$ and $\multiset{a,a,b}$ denote the same multiset.
Given $e \in E$, we denote by $\vec{e}$ the multiset consisting of one occurrence of element
$e$, that is $\multiset{e}$. 
Operations on \(\N\) like addition or comparison are extended to multisets by defining them component wise on each element of \(E\).
Subtraction is allowed as long as each component stays non-negative.
Given a multiset $C$ on $E$  and a multiset $C'$ on $E'$ 
such that $E \cap E'=\emptyset$, we denote by 
$C \cdot C'$ the multiset on $E \cup E'$ equal
to $C$ on $E$ and  to $C'$ on $E'$. 
We call $|C| \defeq\sum_{e\in E} C(e)$ the \emph{size} of $C$, and $\support{C} \defeq \{ e \mid C(e)>0 \}$ the \emph{support} of $C$. 
Given \(E'\subseteq E\) define \(C(E')\defeq\sum_{e\in E'} C(e)\).

\paragraph*{Cubes.}
Given a finite set $Q$, a \emph{cube} $\cube$ is a subset of $\mathbb{M}(Q)$ described 
by a lower bound $L \colon Q \rightarrow \N$ 
and an upper bound $U \colon Q \rightarrow \N \cup \{\infty\}$ 
such that $\cube = \{C : L \le C \le U\}$.
Abusing notation, we identify the set $\cube$ with the pair $(L,U)$.
All the results in this paper are true irrespective of whether the constants
are encoded in unary or binary.

\paragraph*{Reachability.}
Let $\mathcal{T} = (S,\rightarrow)$ be a transition system where $S$ is a set of \emph{configurations} and $\rightarrow$ is a binary relation on $S$
called the transition (or) step
relation.
Given configurations $C$ and $C'$, 
we say $C'$ is \emph{reachable} from $C$ if $C \trans{*} C'$
, where $\trans{*}$ denotes the reflexive-transitive closure of the step relation.
Let $\cSet$ be a set of configurations. 
The \emph{predecessor set} of $\cSet$ is 
$\pre^*_\mathcal{T}(\cSet) \defeq \{ C' | \exists C \in \cSet \, . \, C' \xrightarrow{*} C \}$, and the \emph{successor set} of $\cSet$ is
$\post^*_\mathcal{T}(\cSet) \defeq \{ C | \exists C' \in \cSet \, . \, C' \xrightarrow{*} C \}$.
The \emph{immediate predecessor set} of $\cSet$ is 
$\pre_\mathcal{T}(\cSet) \defeq \{ C' | \exists C \in \cSet \, . \, C' \rightarrow C \}$, and the \emph{immediate successor set} of $\cSet$ is
$\post_\mathcal{T}(\cSet) \defeq \{ C | \exists C' \in \cSet \, . \, C' \rightarrow C \}$.
When it is clear from the context, we will drop the $\mathcal{T}$ subscript.
The \emph{reachability} problem consists of deciding, 
given a system $\mathcal{T}$ and configurations $C,C'$, 
whether $C'$ is reachable from $C$ in $\mathcal{T}$.

\paragraph*{Cube reachability.}

If $\mathcal{T}$ is a transition system whose set of configurations is the set of all multisets on
a finite set $Q$, then the reachability problem can be generalized to the \emph{cube-reachability} problem which consists of deciding, given $\mathcal{T}$ and two cubes $\cube, \cube'$ over $Q$,
whether there exists configurations  $C \in \cube$ and $C' \in \cube'$ such that $C'$ is reachable from $C$ in $\mathcal{T}$.
If this is the case, we say $\cube'$ is reachable from $\cube$.

As mentioned before, the cube-reachability problem generalizes the reachability problem. It also generalizes the coverability problem : Given a configuration $C$ and a state $q \in Q$,
decide if there exists $C'$ such that $C \act{*} C'$ and $C'(q) \ge 1$. It can also talk about \emph{parameterized reachability} problems, for e.g., 
given two finite sets of states $I$ and $F$, do there exist configurations $C$ and $C'$ such that
$\support{C} \subseteq I, \support{C'} \subseteq F$ and $C \act{*} C'$.
Further, the cube-reachability problem is important in the model of immediate observation Petri nets (IO nets). This model was introduced to study immediate observation population protocols \cite{EsparzaGMW18, EsparzaRW19}, and the correctness problem for these protocols is solved using  cube-reachability in IO nets.
Additionally, as we will see in Section~\ref{subsec:leader}, the cube-reachability problem is a generalization of the so-called \emph{leader reachability problem} and allows for an elegant way to transfer results between the models that we study in this paper.


\section{Two Models}

\subsection{Reconfigurable Broadcast Networks}


Reconfigurable broadcast networks (RBN)~\cite{AdHocNetworks,FSTTCS12} are networks comprising an arbitrary number of finite-state, anonymous processes and a communication topology specifying the presence or absence of communication links
between different processes. During a step, a process can broadcast a message which is immediately received by all of its neighbours. 
The process and its neighbours then update their states according to a transition relation. Before each such broadcast step, the communication topology can reconfigure in an arbitrary manner.
Since our main focus in this paper is regarding reachability in this model, we can forget the communication topology and simply define the semantics of an RBN directly in 
terms of collections of processes. 

\begin{definition}
\label{def:rbn}
A \emph{reconfigurable broadcast network} is a tuple 
$\RBN = (Q, \Sigma,\delta)$ 
where $Q$ is a finite set of states,
$\Sigma$ is a finite alphabet 
and $\delta \subseteq Q \times \set{!a,?a \ | \ a \in \Sigma} \times Q$ is the transition relation.
\end{definition}

\begin{figure}
\begin{center}
    \begin{tikzpicture}[->, thick]
      \node[place] (a1) {$a_1$};
      \node[place] (b1) [below =0.5 of a1] {$b_1$};
      \node[place] (c1) [below =0.5 of b1] {$c_1$};
      \node[place] (a2) [right =of a1] {$a_2$};
      \node[place] (b2) [below =0.5 of a2] {$b_2$};
      \node[place] (c2) [below =0.5 of b2] {$c_2$};
      \node[place] (a3) [right =of a2] {$a_3$};
      \node[place] (b3) [below =0.5 of a3] {$b_3$};
      \node[place, ultra thick] (c3) [below =0.5 of b3] {$c_3$};
      \node[place] (tok) [right =2 of a3] {$tok$};
      \node[place] (sent) [right =2 of c3] {$sent$};
      
      \path[->]
      (tok) edge node[right] {$!1$} (sent)
      (a1) edge node[right] {$?1$} (b1)
      (b1) edge node[right] {$?1$} (c1)
      (c1) edge[bend left=40] node[above left] {$!2$} (a1)
      (a2) edge node[right] {$?2$} (b2)
      (b2) edge node[right] {$?2$} (c2)
      (c2) edge[bend left=40] node[above left] {$!3$} (a2)
      (a3) edge node[right] {$?3$} (b3)
      (b3) edge node[right] {$?3$} (c3)
      (c3) edge[bend left=40] node[above left] {$!4$} (a3)
      ;
    \end{tikzpicture}
\end{center}
\caption{An RBN simulating a counter to $2^3$.}
\label{fig:rbn}
\end{figure}
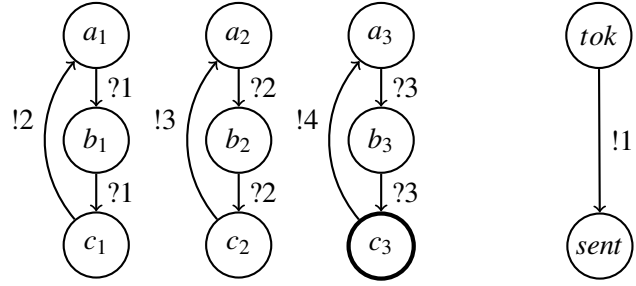

If $(p,!a,q)$ (resp. $(p,?a,q)$) is a transition in $\delta$, we will denote it by $p \act{!a} q$ (resp. $p \act{?a} q$).
A \emph{configuration} $C$ of an RBN $\RBN$ is a multiset over $Q$, which
intuitively counts the number of processes in each state. 
Given a letter $a\in \Sigma$ and two configurations $C$ and $C'$
we say  that there is a \emph{step} $C \trans{a} C'$
if there exists a multiset $\multiset{t, t_1, \ldots, t_k}$ of $\delta$ for some $k\ge 0$
satisfying the following: $t=p \trans{!a} q$, each $t_i =p_i \trans{?a} q_i$,
$C \ge \vec{p} + \sum_i \vec{p_i}$, and $C' = C - \vec{p} - \sum_i \vec{p_i} + \vec{q} + \sum_i \vec{q_i}$. 
We sometimes write this as $C \trans{t+t_1,\ldots, t_n} C'$, 
and intuitively it means that a process at the state $p$ broadcasts the message $a$ and moves to $q$,
and for each $1 \le i \le k$, there is a process at the state $p_i$ which receives this message and moves to $q_i$.
We denote 
by $\trans{*}$ the reflexive and transitive closure of the step relation. 
A \emph{run} is then a sequence of steps.

\begin{example}
\label{ex:rbn}
Consider the RBN of Figure \ref{fig:rbn}, 
with set of states $\set{tok, sent} \cup \set{a_i,b_i,c_i | 1 \le i \le 3}$.
It is inspired by a similar example described in Section 5.1 of \cite{ICALPPatricia}.
Let $\cube_0$ be the cube which puts exactly
one process in each $a_i$, an arbitrary number of processes in $tok$
and $0$ processes elsewhere. 
That is, $\cube_0=(L,U)$ such that $L(a_i)=U(a_i)=1$ for all $i$,
$L(tok)=0$ and $U(tok)=\infty$, and $L(q)=U(q)=0$ for all other states $q$.
Let $\cube_f$ be the cube which puts at least one process in $c_3$ and an arbitrary number elsewhere.
Suppose some configuration in $\cube_0$ reaches some configuration in $\cube_f$.
By construction, for a process to reach $c_3$ it must start in $a_3$ and receive $3$ twice.
For a process to broadcast $3$ it must start in $a_2$ and receive $2$ twice, and for a process to broadcast $2$ it must start in $a_1$ and receive $1$ twice.
So a run from a configuration of $\cube_0$ to a configuration of $\cube_f$
must contain at least $2^3$ broadcasts of $1$.
Since the only way to broadcast $1$ is for a process
to go from $tok$ to $sent$,
 there must be at least $2^3$ processes in $tok$ in 
 the initial configuration of $\cube_0$.
\end{example}
\subsection{Asynchronous Shared-Memory Systems}
Asynchronous shared-memory systems (ASMS)~\cite{JACM16,ModelCheckingSMS}
consist of an arbitrary number of finite-state, anonymous processes.
These processes can communicate with each other by means of a single shared
register, to which they can either write a value or from which they can read a value.

\begin{definition}
	An asynchronous shared-memory system (ASMS) is a tuple $\ASMS = (Q,\Sigma,\delta)$ where
	$Q$ is a finite set of states, $\Sigma$ is a finite alphabet,
	and $\delta \subseteq Q \times \{R,W\} \times \Sigma \times Q$ is the set of transitions.
	Here $R$ stands for \emph{read}, and
	$W$ stands for \emph{write}.
\end{definition}

We use $p \trans{R(d)} q$ (resp. $p \trans{W(d)} q$) to denote
that $(p,R,d,q) \in \delta$ (resp. $(p,W,d,q) \in \delta$). 
The semantics of an ASMS is given by means of \emph{configurations}.
A configuration $C$ of an ASMS is a multiset over $Q \cup \Sigma$
such that $\sum_{d \in \Sigma} C(d) = 1$, i.e., $C$ contains
exactly one element from the set $\Sigma$. 
Hence, we sometimes denote a configuration $C$ as $(M,d)$
where $M$ is a multiset over $Q$ (which counts the number of processes
in each state) and $d \in \Sigma$ (which denotes the content of the shared register).
The value $d$ will be denoted by $\data(C)$.

A \emph{step} between configurations $C = (M,d)$ and $C' = (M',d')$ exists
if there is $t = (p,\oper,d'',q) \in \delta$ such that
$M(p) > 0$, $M' = M - \vec{p} + \vec{q}$
and either $\oper = R$ and $d = d' = d''$ or 
$\oper = W$ and $d' = d''$. If such a step exists,
we denote it by $C \trans{t} C'$ and we let $\trans{*}$ denote
the reflexive transitive closure of the step relation. 
A \emph{run} is then a sequence of steps.
Given a sequence of transitions $\sigma = t_1,\dots,t_n$, we 
sometimes use $C \act{\sigma} C'$ to denote that there is 
a run of the form $C \act{t_1} C_1 \act{t_2} \dots C_{n-1} \act{t_n} C'$.

A cube $\cube = (L,U)$ of an ASMS $\ASMS = (Q,\Sigma,\delta)$ is defined to be a cube over $Q \cup \Sigma$
satisfying the following property : There exists $d \in \Sigma$ such that $L(d) = U(d) = 1$ and $L(d') = U(d') = 0$ for every other $d'$.
Hence, we sometimes denote a cube $\cube$ as $(L,U,d)$ where $(L,U)$ is a cube over $Q$ and $d \in \Sigma$.
Membership of a configuration $C$ in a cube $\cube$ is then defined in a straightforward manner.
The cube-reachability problem for ASMS is then to decide, given $\ASMS$ and two cubes $\cube, \cube'$
whether $\cube$ can reach $\cube'$, i.e., whether there are configurations $C \in \cube, C' \in \cube'$
such that $C \act{*} C'$. 

\begin{figure}
	\begin{center}
    \begin{tikzpicture}[->, thick]
      \node[place] (a1) {$a_1$};
	  \node[place] (a2) [right =of a1] {$a_2$};
	  \node[place] (a3) [right =of a2] {$a_3$};
      \node[place] (a4) [right =of a3] {$a_4$};

      \node[place] (b3) [below =0.75 of a1] {$b_3$};      
      \node[place] (b2) [left =of b3] {$b_2$};
      \node[place] (b1) [left =of b2] {$b_1$};
      
      \node[place] (c1) [below =0.75 of a4] {$c_1$};  
      \node[place] (c2) [right = of c1] {$c_2$};
	  \node[place] (c3) [right = of c2] {$c_3$};

      
      \path[->]
      (a1) edge[bend left = 40] node[above] {$W(1)$} (a2)
      (a1) edge[bend right = 40] node[below] {$W(2)$} (a2)
      (a2) edge node[above] {$R(3)$} (a3)
      (a3) edge node[above] {$R(4)$} (a4)
      
      (b1) edge node[above] {$R(1)$} (b2)
      (b2) edge node[above] {$W(3)$} (b3)
      
      (c1) edge node[above] {$R(2)$} (c2)
      (c2) edge node[above] {$W(4)$} (c3)
      ;
    \end{tikzpicture}
\end{center}
	\caption{An example of an ASMS}
	\label{fig:asms}
\end{figure}
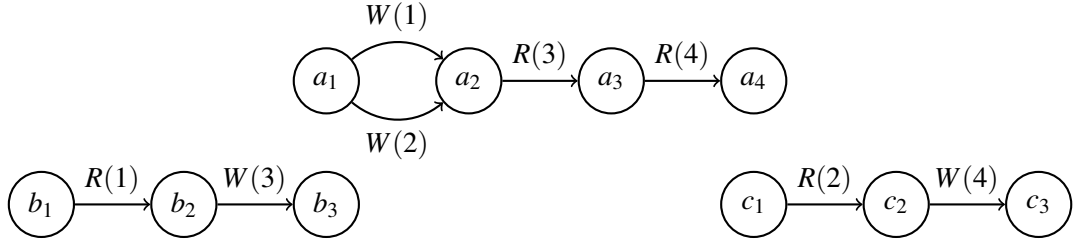

\begin{example}
\label{ex:asms}
Consider the ASMS of Figure~\ref{fig:asms} where the alphabet is $\{\#,1,2,3,4\}$. Let $\cube$ be the cube 
which puts exactly one process in $a_1$, arbitrary number of processes in $b_1$ and $c_1$ and exactly 0 processes elsewhere.
Let $\cube'$ be the cube where $\cube'$ puts at least one process in $a_4$ and arbitrary number of processes elsewhere.
It can be verified that the cube $\cube$ cannot reach $\cube'$ for the following reason: Since there
is only one process in $a_1$ in $\cube$, it follows that this process can either write 1 or 2, but not both.
Hence, either processes from $b_1$ can move into $b_2$ to write 3 or processes from $c_1$ can move into $c_2$ to write 4, but 
both cannot happen. It then follows that it is impossible to read both 3 and 4, and so the state $a_4$ cannot be reached.
\end{example}

\section{RBN and ASMS are Cube-Reachability Equivalent}\label{sec:simulations}


Throughout this paper, whenever we talk about one model simulating another model, we mean that the cube-reachability problem for the second model
can be reduced in polynomial time to the cube-reachability problem for the first model.
In this section, we prove our main result that RBN and ASMS can \emph{simulate} each other.
As we will see in the next section, this simulation will allow us to transfer results from RBN to ASMS and vice versa.


\subsection{ASMS Simulate RBN}\label{subsec:ASMS-RBN}
\textbf{Construction} \
Let $\RBN = (Q_\RBN, \Sigma_\RBN,\delta_\RBN)$ be an RBN. 
We construct an ASMS that simulates $\RBN$.
The register value is used to store which message can be received, additional states are used to represent that a broadcast is in progress, and a fresh register value is written when the simulation of a broadcast is over.
\chana{high level intuition added}
For every $a \in \Sigma_\RBN$, we let $\delta_\RBN^{!a}$ (resp. $\delta_\RBN^{?a}$) be the subset 
of the transitions in $\delta_\RBN$ that broadcast (resp. receive) the letter $a$.
Let $\ASMS = (Q_\ASMS,\Sigma_\ASMS,\delta_\ASMS)$ be the following ASMS: 
The set of states $Q_\ASMS$ is $Q_\RBN \cup I$ with $I=\set{[p,a,p'] : (p,?a,p') \in \delta \text{ or } (p,!a,p') \in \delta}$, where $I$ stands for intermediary.
The alphabet $\Sigma_\ASMS$ is $\Sigma_\RBN \cup \{\#\}$ where $\#$ is a letter which is not in  $\Sigma$.
The transition relation $\delta_\ASMS$ is such that 
for every $t = (q,!a,q') \in \delta_\RBN$ 
there are transitions $\hat{t} := q \trans{W(a)} [q,a,q']$ and $t^\# := [q,a,q'] \trans{W(\#)} q'$ in $\delta_\ASMS$, and
for every $t = (q,?a,q') \in \delta_\RBN$ 
there are transitions $\hat{t} := q \trans{R(a)} [q,a,q']$ and $t^\# := [q,a,q'] \trans{W(\#)} q'$ in $\delta_\ASMS$,
as represented in Figure \ref{fig:rbn-to-sms}. \chana{fig 3 added and mentioned}

\begin{figure}[h]
\begin{center}
    \begin{tikzpicture}[->, thick, node distance=1.25cm] 
      \node[place] (q) {$q$}; 
      \node[place,inner sep=1pt] (qaq') [right =of q] {$q,a,q'$};
      \node[place] (q') [right =of qaq'] {$q'$};
      \node[place] (p) [right =of q'] {$p$};
      \node[place,inner sep=1pt] (pap') [right =of p] {$p,a,p'$};
      \node[place] (p') [right =of pap'] {$p'$}; 
      
      \path[->]
      (q) edge node[above] {$W(a)$} (qaq')
      (p) edge node[above] {$R(a)$} (pap')
      (pap') edge node[above] {$W(\#)$} (p')
      (qaq') edge node[above] {$W(\#)$} (q')
      ;
    \end{tikzpicture}
\end{center}
\caption{Simulation in $\ASMS$ of transitions $q\trans{!a}q'$ and $p \trans{?a}p'$ of $\RBN$.}
\label{fig:rbn-to-sms}
\end{figure}

A configuration $(C,d)$ of $\ASMS$ is called \emph{good} if
$C(I)=0$ and $d=\#$. 
There is a natural bijection between configurations of $\RBN$ and good configurations of $\ASMS$. 
If $C$ is a configuration of $\RBN$, we will use $(\widehat{C},\#)$ to denote the corresponding good configuration of $\ASMS$.

\noindent \textbf{Correctness of construction} \ 
We now show that $C' \in \poststar_{\RBN}(C)$ iff $(\widehat{C'},\#) \in \poststar_{\ASMS}(\widehat{C},\#)$ for any configurations $C$ and $C'$ of $\RBN$.
Suppose $C \trans{t + t_1,\dots,t_n} C'$ is a step in $\RBN$.
It is easy to see that we have a run in $\ASMS$ of the form $(\widehat{C},\#) \act{\hat{t},\hat{t_1},
\dots,\hat{t_n},t^\#,t_1^\#,\dots,t_n^\#} (\widehat{C'},\#)$.
Hence, if $C \trans{*} C'$ for some configurations $C,C'$ in $\RBN$,
then $(\widehat{C},\#) \trans{*} (\widehat{C'},\#)$ in $\ASMS$.

For the other direction, we first define the notion of a \emph{pseudo-step} between two good configurations of $\ASMS$. 
A run $(\widehat{C},\#) \act{\sigma} (\widehat{C'},\#)$ of $\ASMS$ 
is called a \emph{pseudo-step} if there exists $a \in \Sigma$
and transitions $t \in \delta_{\RBN}^{!a}$ and $t_1,\dots,t_n \in \delta_{\RBN}^{?a}$ 
such that $\sigma = \hat{t},\hat{t_1},\dots,\hat{t_n},t^\#,t_1^\#,\dots,t_n^\#$.
The intuition behind this notion is that if $(\widehat{C},\#) \act{\sigma} (\widehat{C'},\#)$ 
where $\sigma$ is a pseudo-step with $\sigma = \hat{t},\hat{t_1},\dots,\hat{t_n},t^\#,t_1^\#,\dots,t_n^\#$
then $C \act{t+t_1\dots t_n} C'$ is a step in $\RBN$. Hence, pseudo-steps of $\ASMS$ ``behave'' similarly to a single step in $\RBN$.

Now a run $(\widehat{C},\#) \act{\sigma} (\widehat{C'},\#)$ of $\ASMS$ is said to be in \emph{normal form} if it is either the empty run or if it can be decomposed into a sequence of \emph{pseudo-steps}.
Hence, it follows that if $(\widehat{C},\#) \act{\sigma} (\widehat{C'},\#)$
is a run in normal form then $C \act{*} C'$ in $\RBN$. 

The following lemma asserts that whenever there is a run between two good configurations of 
$\ASMS$, then there is also a run between those configurations in normal form. Hence, using this
lemma and the discussion in the previous paragraph, it follows that if $(\widehat{C},\#) \act{*} (\widehat{C'},\#)$ in $\ASMS$ then $C \act{*} C'$ in $\RBN$.


\begin{restatable}[Normal form lemma]{lemma}{LmNormalASMS}
    Suppose $(\widehat{C},\#) \act{\rho} (\widehat{C'},\#)$ is a run in $\ASMS$.
    Then there exists  $\sigma$ such that $(\widehat{C},\#) \act{\sigma} (\widehat{C'},\#)$
    is a run in normal form.
\end{restatable}

\begin{proof}[Proof sketch of normal form lemma]
Let $n$ be  the length of $\rho$. \chana{correcting induction start. is it ok that length of sequences is not defined?}
We proceed by induction on $n$. If $n = 0$, we are done.
Let $n > 0$ and $\rho = \rho_1,\dots,\rho_n$. Assume now that any run of length strictly less than $n$ can be put in normal form. By analysing the
structure of the transitions in $\ASMS$ and noticing that $\rho$ begins at a good configuration,
we can first show that $\rho_1$ must be of the form $q \act{W(a)} [q,a,q']$ for some $a \in \Sigma$. 
Then we consider two cases:

\textbf{Case 1: } Suppose there is no $i > 1$ such that $\rho_i$ is a transition which
writes a value $b \neq \#$.  Hence, every transition in $\rho_2,\dots,\rho_n$ either reads the value $a$ or writes $\#$ and so
there must be an index $2 \le j \le n$ such that every transition in $\rho_2,\dots,\rho_{j-1}$
reads $a$ and every transition in $\rho_j, \dots, \rho_n$ writes $\#$. Now, by analysing
the transitions going in and out of the subset $I$ and noticing that the run
begins and ends at good configurations, we can show that
$(\widehat{C},\#) \act{\rho} (\widehat{C'},\#)$ must be a pseudo-step.

\textbf{Case 2: } Suppose there is $i > 1$ such that $\rho_i$ is a transition which writes
a value $b \neq \#$. By the same argument as before, it is easy to see that 
there must exist $2 \le j \le i-1$ such that every transition in $\rho_2,\dots,\rho_{j-1}$
reads $a$ and every transition in $\rho_j, \dots, \rho_{i-1}$ writes $\#$. Let $Z$
be the configuration reached after $\rho_{i-1}$.
Let $M = Z(I)$, i.e., $M$ is the multiset of processes at the configuration $Z$
which are in some intermediary state. Since the only way out of the set $I$ is to
write $\#$ onto the register, if $M = \multiset{[p_1,a,p_1'],\dots,[p_k,a,p_k']}$ then
there must exist $i_1,\dots,i_k > i$ such that each $\rho_{i_l}$ is $[p_l,a,p_l'] \act{W(\#)} p_l'$. 
We can then rearrange the run by first following $\rho$ up till $\rho_{i-1}$, then ``preponing'' the
transitions $\rho_{i_1},\dots,\rho_{i_k}$ and then firing the rest of $\rho$ to reach $(\widehat{C'},\#)$.
With this rearrangement, the run up till $\rho_{i_k}$ becomes a pseudo-step
and so we can apply induction hypothesis on the rest of the run.
\end{proof}

\noindent \textbf{The reduction} \ With this construction, we can now simulate RBN by ASMS as follows: Let $\RBN$ be an RBN with states $Q_\RBN$ and let $\cube_1 = (L_1,U_1),\cube_1' = (L_1',U_1')$ be two cubes of $\RBN$. Construct the ASMS $\ASMS$ as described above.
Then construct the following two cubes $\cube_2 = (L_2,U_2,\#), \cube_2' = (L_2',U_2',\#)$ of $\ASMS$:
$L_2(q), U_2(q), L_2'(q)$ and $U_2'(q)$ are respectively equal to $L_1(q), U_1(q), L_1'(q)$ and $U_1'(q)$ if $q$ is a state of $\RBN$.
If $q$ is in $I$, then $L_2(q) = U_2(q) = L_2'(q) = U_2'(q) = 0$.
It is then easy to see that $C \in \cube_1$ (resp. $\cube_1'$) iff $\hat{C} \in \cube_2$ (resp. $\cube_2'$). Hence, by correctness of our construction, it follows that $\cube_1$ can reach $\cube_1'$ iff $\cube_2$ can reach $\cube_2'$.

\subsection{RBN Simulate ASMS}\label{subsec:RBN-ASMS}


\textbf{Construction} \ Let $\ASMS = (Q_{\ASMS},\Sigma,\delta_{\ASMS})$ be an ASMS. 
We construct an RBN $\RBN$ where one agent  acts like the register of $\ASMS$ and
all the other agents behave like agents of $\ASMS$.
Let $\RBN = (Q_{\RBN},\Sigma_{\RBN},\delta_{\RBN})$ be an RBN defined as follows:
The set of states $Q_{\RBN}$ is comprised of two parts.
The first part consists of the set $Q_{\ASMS} \cup \{[p,a,q]: p \trans{W(a)} q \in \delta_{\ASMS}\}$,
which will  intuitively be used to simulate the processes of $\ASMS$.
The second part consists of the set $\Sigma \cup \{\overline{a} : a \in \Sigma\}$
which will  intuitively be used to simulate the register of $\ASMS$.
The set $\{\overline{a} : a \in \Sigma\}$ is denoted by $\overline{\Sigma}$.
The alphabet $\Sigma_{\RBN}$ is $\{Read_a, Ch_a, Ack_a : a \in \Sigma\}$.

Before describing the transition relation $\delta_{\RBN}$ we set up some notation:
A \emph{good} configuration of $\RBN$ is a configuration $C$ 
such that $\sum_{a \in \Sigma} C(a) = 1$ and
$C(p) = 0$ if $p \notin Q_{\ASMS} \cup \Sigma$.
Intuitively, in a good configuration, there is one process which stores
the value of the register of $\ASMS$ and all the
other processes are in some state of $Q_{\ASMS}$.
Notice that there is a natural bijection between configurations of $\ASMS$ and good configurations of $\RBN$.
If $C$ is a configuration of $\ASMS$, we will use $\widehat{C}$ to denote the corresponding good configuration of $\RBN$.

\begin{figure}
\begin{center}
    \begin{tikzpicture}[->, thick, node distance=1.25cm] 
      \node[place] (q) {$q$}; 
      \node[place,inner sep=1pt] (qaq') [right =of q] {$q,a,q'$};
      \node[place] (q') [right =of qaq'] {$q'$};
      \node[place] (d) [right =of q'] {$d$};
      \node[place] (abar) [right =of d] {$\overline{a}$};
      \node[place] (a) [right =of abar] {$a$}; 
      \node[place] (p) [right =of a] {$p$};
      \node[place] (p') [right =of p] {$p'$}; 
      
      \path[->]
      (q) edge node[above] {$?Ch_a$} (qaq')
      (qaq') edge node[above] {$!Ack_a$} (q')
      (p) edge node[above] {$?Read_a$} (p')
      (d) edge node[above] {$!Ch_a$} (abar)
      (a) edge [loop above] node {$!Read_a$} (a)
      (abar) edge node[above] {$?Ack_a$} (a)
      ;
      
    \end{tikzpicture}
\end{center}
\caption{Simulation in $\RBN$ of transitions $q\trans{W(a)}q'$ and $p \trans{R(a)}p'$ of $\ASMS$.}
\label{fig:sms-to-rbn}
\end{figure}

Now, the transition relation $\delta_{\RBN}$ is constructed so that the following invariant is satisfied:
For any configurations $C$ and $C'$ of $\ASMS$, $C' \in \poststar_{\ASMS}(C)$ iff $\widehat{C'} \in \poststar_{\RBN}(\widehat{C})$.

\begin{itemize}
	\item Suppose $t = p \trans{R(a)} q$ is a transition in $\ASMS$. Correspondingly, we have 
	two transitions $a \trans{!Read_a} a$ and $p \trans{?Read_a} q$ in $\RBN$.
	Hence, if $C \trans{t} C'$ in $\ASMS$, then $\widehat{C} \trans{(a,!Read_a,a) + (p,?Read_a,q)} \widehat{C'}$ in $\RBN$.
	\item Suppose $t = p \trans{W(a)} q$ is a transition in $\ASMS$. 
	We first have two transitions $p \trans{?Ch_a} [p,a,q]$ and $[p,a,q] \trans{!Ack_a} q$.
	Further, for \textbf{every} $d \in \Sigma$, we have the transitions, $d \trans{!Ch_a} \overline{a}$
	and $\overline{a} \trans{?Ack_a} a$.
	Intuitively, the process responsible for the register requests to \emph{change} the value of the register from $d$ to $a$ by 
	broadcasting the message $Ch_a$ and moving to $\overline{a}$.
	The process at state $p$ is capable of receiving this message and moves to the state $[p,a,q]$ 
	and from there it is capable of sending the message $Ack_a$ \emph{acknowledging} the change sent by the register.
	The process at $\overline{a}$ can receive $Ack_a$ and move to $a$.
	Hence, if $C \trans{t} C'$ then $\widehat{C} \trans{(d,!Ch_a,\overline{a}) + (p,?Ch_a,[p,a,q])} C_{int} \trans{([p,a,q],!Ack_a,q) + (\overline{a},?Ack_a,a)} \widehat{C'}$.
	Figure \ref{fig:sms-to-rbn} represents the transitions needed for this simulation.\chana{ref to fig 4}
\end{itemize}

\noindent \textbf{Correctness of construction} \
Hence, if $C \trans{*} C'$ in $\ASMS$ then we have shown that $\widehat{C} \trans{*} \widehat{C'}$ in $\RBN$. Notice
that we have also shown that it is possible to go from $\widehat{C}$ to $\widehat{C'}$ where every broadcasted message is received by \emph{exactly} one other process.
Our next lemma shows that this is not an accident, and indeed any run between $\widehat{C}$ and $\widehat{C'}$ can be transformed into this form.

A run between good configurations of $\RBN$ is said to be in \emph{normal form} if whenever $Z \trans{t + t_1,\dots,t_n} Z'$ is a step
in that run, then $n = 1$. We have the following lemma.
\begin{restatable}[Normal form lemma]{lemma}{LmNormalRBN}
	Suppose there is a run from $Z$ to $Z'$ in $\RBN$ where $Z$ and $Z'$ are good configurations.
	Then there is a run from $Z$ to $Z'$ which is in normal form.
\end{restatable}

First we will see how our simulation is correct, using the normal form lemma.
Suppose $\widehat{C} \trans{*} \widehat{C'}$ in $\RBN$ for some configurations $C$ and $C'$ of $\ASMS$.
By the normal form lemma, we can assume that this run is in normal form and so 
let $\widehat{C} \trans{b^1 + r^1} C_1 \trans{b^2 + r^2} C_2 \dots C_{m-1} \trans{b^m + r^m} \widehat{C'}$.
We proceed by induction on $m$. The base case of $m = 0$ is trivial.
Suppose $m > 0$ and assume the claim holds for all numbers less than $m$.
Since $\widehat{C}$ is a good configuration, there are only two possible cases for $b^1$:

\textbf{Case 1: } Suppose $b^1 = (a,!Read_a,a)$ for some $a \in \Sigma$. 
Hence $r^1$ must be $(p,?Read_a,q)$ for some $p, q \in Q_{\ASMS}$. 
It follows that $C_1 = \widehat{Z}$ for some configuration $Z$ of $\ASMS$.
Since $C \trans{(p,R,a,q)} Z$ in $\ASMS$, by applying the induction hypothesis on the run from $\widehat{Z}$ to $\widehat{C'}$,
we are done.

\textbf{Case 2: } Suppose $b^1 = (d,!Ch_a,\overline{a})$ for some $d,a \in \Sigma$.
Hence $r^1$ must be $(p,?Ch_a,[p,a,q])$ for some $p, q \in Q_{\ASMS}$.
The only process
which can broadcast from $C_1$ is the process at $[p,a,q]$ 
and moreover it can only broadcast $Ack_a$. 
The only process which can receive $Ack_a$ from $C_1$
is the process at the state $\overline{a}$. Hence $b^2 = ([p,a,q],!Ack_a,q)$ and $r^2 = (\overline{a},?Ack_a,a)$.
Therefore, $C_2 = \widehat{Z}$ for some configuration $Z$ of $\ASMS$.
Since $C \trans{(p,W,a,q)} Z$ in $\ASMS$, by applying the induction hypothesis on the run from
$\widehat{Z}$ to $\widehat{C'}$, we are done.

\begin{proof}[Proof sketch of normal form lemma]
Suppose $Z_0 := Z \trans{b^1 + r_1^1,\dots,r_{n_1}^1} Z_1 \trans{b^2 + r_1^2,\dots,r_{n_2}^2} Z_2 \dots Z_{m-1} \trans{b^m + r_1^m,\dots,r_{n_m}^m} Z_m := Z'$.
We proceed by induction on $m$. The case of $m = 0$ is trivial.
	
Suppose $m > 0$ and assume that the claim is true for all numbers less than $m$. Since $Z_0$ is a good configuration, there are only two possible choices for $b^1$.
	
\textbf{Case 1: } Suppose $b^1 = a \trans{!Read_a} a$ for some $a \in \Sigma$.
By firing $b^1$ repeatedly, we can fire $r^1_1,r^1_2,\dots,r^1_{n_1}$ ``one at a time'' and reach
$Z_1$ from $Z_0$ using a run in normal form. We can then apply the induction hypothesis on the run between $Z_1$ and $Z'$.

\textbf{Case 2: } Suppose $b^1 = d \trans{!Ch_a} \overline{a}$ for some $d,a \in \Sigma$.
Hence, $Z_1$ is a bad configuration and so $Z_1 \neq Z'$. If $n_1 = 0$, then no process in $Z_1$ can broadcast any message,
which leads to a contradiction. So, $n_1 > 0$.

For each $1 \le i \le n_1$, let $r_i^1 = (p_i,?Ch_a,[p_i,a,q_i])$.
Let $S := \sum_{i=2}^{n_1} \vec{p_i} - \sum_{i=2}^{n_1} \vec{[p_i,a,q_i]}$
and let $M := \sum_{i=2}^{n_1} \vec{[p_i,a,q_i]}$.
Notice that the only processes which can broadcast a message at the configuration $Z_1$
are the processes in the multiset $\vec{[p_1,a,q_1]} + M$. 
Hence $b^2 = (p_i[a]q_i,!Ack_a,q_i)$ for some $i$. Without loss of generality, we can assume that $i = 1$.

Notice that the only process which can receive the message $Ack_a$ at the configuration $Z_1$
is the process at the state $\overline{a}$. It then follows that either $n_2 = 0$ or $n_2 = 1$. 
Hence, we get two subcases:

\textbf{Case 2a): } Suppose $n_2 = 0$. 
Then reorder the run between $Z_0$ and $Z_2$ as follows:
$Z_0 \trans{b^1 + r_1^1} Z_1 + S \trans{b^2 + (\overline{a},?Ack_a,a)} Z_2 + S - \overline{\vec{a}} + \vec{a} 
\trans{(a,!Ch_a,\overline{a}) + r_2^1,\dots,r_{n_1}^1} Z_2$. Notice that the 
configuration $Z_2 + S - \overline{\vec{a}} + \vec{a}$ is a good configuration and 
has a run of length $m-1$ to $Z'$. Applying induction hypothesis, we are then done.

\textbf{Case 2b): } Suppose $n_2 = 1$. Hence, $r_1^2 = (\overline{a},?Ack_a,a)$ and so
$Z_2(\overline{a}) = 0$ and $Z_2(a) = 1$.
We consider two further subcases:
	\begin{itemize}
		\item Suppose there exists $\alpha > 2$ such that $Z_{\alpha}(\overline{a}) = 1$.
		Let $\alpha$ be the minimum such index. Hence, there must exist some $d' \in \Sigma$ such that $b^{\alpha}$ is $(d',!Ch_a,\overline{a})$.

		Suppose no $b \in \{b^3,\dots,b^{\alpha-1}\}$ is labelled by $!Ack_a$. 
		Intuitively, we can then show that none of the processes in any of the states
		in the multiset $M$ ever make a step between $Z_2$ and $Z_\alpha$. 
		Hence, we can ``postpone'' firing the transitions $r^1_2,\dots,r^1_{n_1}$ and get
		$Z_0 \trans{b^1 + r_1^1} Z_1 + S \trans{b^2 + r_1^2} Z_2 + S \act{*} Z_{\alpha-1}+S \act{b^\alpha + r^1_2,\dots,r^1_{n_1},r_1^{\alpha},\dots,r_{n_\alpha}^\alpha} Z_\alpha$. The configuration $Z_2 + S$ is a good configuration and has a run to $Z'$ of length $m-2$ and so we can apply the induction hypothesis.
		
		Suppose some $b \in \{b^3,\dots,b^{\alpha-1}\}$ is labelled by $!Ack_a$. Let $b = b^i$ be the first such transition.
		By definition of $\alpha$ and by construction of the protocol, we can show that $b^i$ must be $([p_j,a,q_j],!Ack_a,q_j)$ for some $2 \le j \le n_1$ (without loss of
		generality we can assume $j = 2$) and we can also show that no process at the $i^{th}$ step receives this message, i.e. $n_i = 0$.
		Hence, we can ``prepone'' firing the transition $b^i$ and get $Z_0 \trans{b^1 + r_1^1} Z_1 + S \trans{b^2 + r_1^2} Z_2 + S
		\trans{(a,!Ch_a,\overline{a}) + r_2^1,\dots,r_{n_1}^1} Z_2 - \vec{a} + \overline{\vec{a}} 
		\trans{(b^i + (\overline{a},?Ack_a,a)} Z_2 - \vec{[p_2,a,q_2]} + \vec{p_2} \act{*} 
		Z_{i-2} - \vec{[p_2,a,q_2]} + \vec{p_2} \trans{b^{i-1} + r_1^{i-1} \dots r_{n_{i-1}}^{i-1}} Z_i$.
		Notice that $Z_2 + S$ is a good configuration and has a run to $Z'$ of length $<m$ and so we can apply the induction hypothesis.
		
		\item Suppose there does not exist $\alpha > 2$ such that $Z_{\alpha}(\overline{a}) = 1$.
		
		Suppose no $b \in \{b^3,\dots,b^{m}\}$ is labelled by $!Ack_a$. We can once again show that none of the processes in the multiset $M$ 
		ever make a step between $Z_2$ and $Z_m$. Since $Z_m$ is a good configuration, it must then be the case that $n_1 = 1$, which means that
		$Z_2$ is a good configuration and the run between $Z_0$ and $Z_2$ is already in normal form. Because of the induction hypothesis, we are done.
		 
		Suppose some $b \in \{b^3,\dots,b^{m}\}$ is labelled by $!Ack_a$. Let $b = b^i$ be the first such transition.
		In this case, we can do a similar rearrangement like the corresponding previous case by ``preponing'' $b^i$ 
		and then conclude by applying the induction hypothesis.
	\end{itemize}
\end{proof}

\noindent \textbf{The reduction} \
Now, suppose we are given an ASMS $\ASMS$ and two cubes $\cube_1 = (L_1,U_1,d)$ and $\cube'_1 = (L'_1,U'_1,d')$. 
We construct the protocol $\RBN$ as we have described in this section. 
Then we construct two cubes $\cube_2 = (L_2,U_2)$ and $\cube'_2 = (L_2',U_2')$ of $\RBN$ as follows: 
$L_2(q), U_2(q), L_2'(q)$ and $U_2'(q)$ are all respectively equal to $L_1(q), U_1(q), L_1'(q)$ and $U_1'(q)$ if $q \in Q_\ASMS$, 
$L_2(d) = U_2(d) = L_2'(d') = U_2'(d') = 1$ and otherwise $L_2(q) = U_2(q) = L_2'(q) = U_2'(q) = 0$.
It is easy to see that a configuration $C \in \cube_1$ (resp. $\cube'_1$) iff its corresponding 
configuration $\hat{C} \in \cube_2$ (resp. $\cube'_2$). Hence, by our simulation it follows that
$\cube_1$ can reach $\cube'_1$ in $\ASMS$ iff $\cube_2$ can reach $\cube'_2$ 
in $\RBN$.\\

\noindent \textbf{Another reduction} \
While this construction proves the desired result, we need a slightly different construction for the purposes of the next section which we now describe. Given an ASMS $\ASMS$ and 
two cubes $\cube_1 = (L_1,U_1,d)$ and $\cube'_1 = (L_1',U_1',d')$, once again construct the RBN $\RBN$ described in this section
and construct two cubes $\cube_3 = (L_3,U_3)$ and $\cube'_3 = (L_3',U_3')$ of $\RBN$ as follows:
The cube $\cube'_3$ is the same as $\cube'_2$ described before. The cube $\cube_3$ is also
exactly the same as $\cube_2$, except for the constraints $L_2(d) = U_2(d) = 1$ which
are replaced by $L_3(d) = 0, U_3(d) = \infty$.

Since $\cube_2 \subseteq \cube_3$, it follows from the previous reduction that if $\cube_1$ can reach $\cube_1'$,
then $\cube_3$ can reach $\cube_3'$.
For the other direction, notice that, by construction of the protocol $\RBN$,
\begin{equation}\label{eq:invariant}
	\text{If } C \act{} C' \text{ is a step in } \RBN,
	\text{ then } \sum_{q \in \Sigma \cup \overline{\Sigma}} C(q) = \sum_{q \in \Sigma \cup \overline{\Sigma}} C'(q)
\end{equation}
Using this equation and the fact that any configuration in $\cube_3'$ is a good configuration,
it is then clear that if $C_3 \in \cube_3$ such that $C_3 \act{*} C_3'$ with $C_3' \in \cube_3'$,
then $C_3$ must also be a good configuration. 
Hence, we can then conclude 
that $\cube_3$ can reach $\cube_3'$ iff $\cube_1$ can reach $\cube_1'$.

\section{Transferring Existing Results}\label{sec:consequences}
In the previous section, we have shown that RBN and ASMS are polynomial-time equivalent with respect \chana{updated} to the cube-reachability problem. 
Though the precise complexity of this 
problem has not been established for either one of these models, our result shows that it is sufficient to characterize the complexity of cube-reachability
for one of these models. 
Moreover, there exist results for subclasses \chana{updated}
of the cube-reachability problem for both RBN and ASMS. In this section, we use the reductions constructed in the previous section to transfer these results from RBN to ASMS and vice versa. 

\subsection{Unbounded initial cube reachability}
\label{subsec:unbounded-cube}

We consider the following problem for RBN, which we call the \emph{unbounded initial cube reachability} problem: We are given an RBN $\RBN$ and two cubes $\cube = (L,U), \cube' = (L',U')$ with the special property that $L(q) = 0$ and $U(q) \in \{0,\infty\}$ for every state $q$ and 
we would like to check if $\cube$ can reach $\cube'$. Notice that there is no restriction on the cube $\cube'$. 
We will call such a pair $(\cube,\cube')$ as an \emph{unbounded initial cube} pair.
This problem was proved to be \textsf{PSPACE}-complete for RBN in (\cite{FSTTCS12}, Theorem 5.5).
(In~\cite{FSTTCS12}, this result is only stated for cubes with constants encoded in unary, but the proof
can be modified easily to also give the same upper bound when the constants are encoded in binary).

In a similar way, it is possible to define the corresponding problem for ASMS. 
Notice that if $(\cube,\cube')$ is an unbounded initial cube pair for an ASMS $\ASMS$, then
the second reduction in Section~\ref{subsec:RBN-ASMS}  produces an RBN $\RBN$ along with an unbounded initial cube pair as well. This shows that the corresponding problem for ASMS can be solved in \textsf{PSPACE}. 

Further, notice that given an RBN $\RBN$ and an unbounded initial cube pair for $\RBN$, our reduction in Section~\ref{subsec:ASMS-RBN} produces an ASMS $\ASMS$ with an unbounded initial cube pair as well. This shows that the unbounded initial cube reachability problem for ASMS is \textsf{PSPACE}-hard.

\begin{theorem}
	The unbounded initial cube reachability problem for ASMS is \textsf{PSPACE}-complete.
\end{theorem}

\subsection{Leader protocols}\label{subsec:leader}
The notion of an ASMS equipped with a \emph{leader} has been studied in~\cite{JACM16, FineGrained}. Formally, an ASMS-leader protocol is a pair of ASMS protocols $\ASMS_C = (Q_C,\Sigma,\delta_C), 
\ASMS_D = (Q_D,\Sigma,\delta_D)$, 
where $\ASMS_C$ is called the \emph{contributor protocol} and $\ASMS_D$ is called the \emph{leader protocol}.
Intuitively, there is exactly one process which executes $\ASMS_D$ (the leader) and all the other processes execute $\ASMS_C$ (contributors).
This is formalized as follows: A configuration of such a system is defined to be a triple $(q,M,a)$ where $q \in Q_D$, $M$ is a multiset on $Q_C$
and $a \in \Sigma$. A step between $C = (q,M,a)$ and $C' = (q',M',a')$ exists if one of the following is true:
\begin{itemize}
	\item There exists $(q,\oper,a',q') \in \delta_D$ such that $M' = M$ and either $\oper = R$ and $a = a'$, or $\oper = W$.
	\item There exists $(p,\oper,a',p') \in \delta_C$ such that $q = q'$, $M(p) \ge 1$, $M' = M - \vec{p} + \vec{p'}$, and either $\oper = R$ and $a = a'$, or $\oper = W$.
\end{itemize}

We can then define the notion of a run for an ASMS-leader protocol in the usual way. The \emph{ASMS-leader reachability} problem is to decide, given an ASMS-leader protocol $(\ASMS_C,\ASMS_D)$, two leader states $q_D^I, q_D^f$, a contributor state $q_C^I$ and two data values $a,a' \in \Sigma$ whether there exists a $k \ge 1$ such that the configuration $(q_D^I,\multiset{k \cdot q_C^I},a)$ can reach a configuration $C' = (q_D^f,M',a')$ for some $M'$.

We now define a special case of cube-reachability in ASMS and notice that this special case is exactly equivalent to ASMS-leader reachability.
An \emph{ASMS-leader cube} is a pair $(\ASMS,\cube,\cube')$ of the following form: The protocol $\ASMS = (Q,\Sigma,\delta)$ is such that there exists a partition of the states and transition relation as $Q = Q_C \cup Q_D, \delta = \delta_C \cup \delta_D$ 
and $\cube = (L,U,a), \cube' = (L',U',a')$ satisfy: There exists exactly two states $q_D^I, q_D^f \in Q_D$ such that $L(q_D^I) = U(q_D^I) = 1, L'(q_D^f) = U'(q_D^f) = 1$ 
and for every other state $q \in Q_D$, $L(q) = L'(q) = U(q) = U'(q) = 0$ and there exists exactly one state $q_C^I \in Q_C$ such that $L(q_C^I) = L'(q_C^I) = 0, U(q_C^I) = U'(q_C^I) = \infty$ and for every other state $q \in Q_C$, $L(q) = U(q) = L'(q) = 0, U'(q) = \infty$. 
Notice that Example~\ref{ex:asms} is an example of an ASMS-leader cube.

It is easy to see that the ASMS-leader reachability problem is equivalent to the cube-reachability problem for ASMS-leader cubes.
The following result has been shown for ASMS.
\begin{theorem}[\cite{JACM16}]
	The ASMS-leader reachability problem is in \textsf{NP}.
\end{theorem}

Now, we can define RBN-leader protocols and RBN-leader cubes in exactly the same way as was done for ASMS. Further, notice that the reduction given in Section~\ref{subsec:ASMS-RBN} has the following special property:
If we are given an RBN-leader cube $(\RBN,\cube_1,\cube'_1)$, then the reduction produces an ASMS-leader cube $(\ASMS,\cube_2,\cube_2')$. This proves that the RBN-leader cube reachability problem (and hence the RBN-leader
reachability problem) is in \textsf{NP}. 

Notice that the reduction given in Section~\ref{subsec:RBN-ASMS} does not output a RBN-leader cube when it is given an ASMS-leader cube as input. Hence, we do not immediately get \textsf{NP}-hardness of the 
RBN-leader reachability problem. Nevertheless, by a reduction from 3-SAT similar to that of the one given in Theorem 10 of~\cite{FineGrained}, we can prove \textsf{NP}-hardness of RBN-leader reachability. We then get
\begin{restatable}{theorem}{ThmLeaderRBN}
	The RBN-leader reachability problem is \textsf{NP}-complete.
\end{restatable}


\subsection{Almost-sure coverability}

We now consider the notion of \emph{almost-sure coverability} for ASMS. Let $\ASMS = (Q,\Sigma,\delta)$ be an ASMS with two distinguished states $q_I, q_f$ and 
a distinguished letter $d \in \Sigma$. Let $\uparrow q_f$ denote the set of all configurations $C$ such that $C(q_f) \ge 1$.
For any $k \ge 1$, we say that the configuration $(\multiset{k \cdot q_I},d)$ \emph{almost-surely covers} $q_f$ iff $\post^*((\multiset{k \cdot q_I},d)) \subseteq \pre^*(\uparrow q_f)$. 
The reason behind calling this the almost-sure coverability relation is that the definition given
here is equivalent to covering the state $q_f$ from $(\multiset{k \cdot q_I},d)$ with probability 1
under a probabilistic scheduler which picks processes uniformly at random at each step.

The number $k$ is called a \emph{cut-off} if one of the following is true: 1) Either for all $h \ge k$, the configuration $(\multiset{h \cdot q_I},d)$ almost-surely covers $q_f$. In this case, $k$ is a positive cut-off. Or, 2) for all $h \ge k$, the configuration $(\multiset{h \cdot q_I},d)$ does not almost-surely cover $q_f$. In this case, $k$ is a negative cut-off.
Note that from the definition alone, it is not clear that a cut-off must exist for every ASMS. The following result is known.
\begin{theorem}[Theorem 3 of~\cite{ICALPPatricia}]~\label{thm:ICALPPatricia}
	Given an ASMS with two states $q_I,q_f$ and a letter $d$, a cut-off always exists. Whether the cut-off is positive or negative can be decided in \textsf{EXPSPACE}.
\end{theorem}

We can now translate this result to RBNs. Given an RBN $\RBN$ and two states $q_I, q_f$, we first set $\uparrow q_f := \{C : C(q_f) \ge 1 \}$.
Then for any $k \ge 1$, we say that $\multiset{k \cdot q_I}$ almost-surely covers $q_f$ iff $\post^*(\multiset{k \cdot q_I}) \subseteq \pre^*(\uparrow q_f)$. 
We can then define positive and negative cut-offs in a similar manner. 
Now for the RBN $\RBN$, let $\ASMS$ be the ASMS protocol that we construct in our reduction given in Section~\ref{subsec:ASMS-RBN}. Using the construction of $\ASMS$, we can then easily show that
\begin{quote}
    for any $k \ge 1$, $\post^{*}_{\RBN}(\multiset{k \cdot q_I}) \subseteq \pre^{*}_{\RBN}(\uparrow q_f)$ iff $\post^{*}_{\ASMS}(\multiset{k \cdot q_I},\#) \subseteq \pre^{*}_{\ASMS}(\uparrow q_f)$.    
\end{quote}
This then directly implies that
\begin{restatable}{theorem}{ThmASCoverabilityRBN}
	Given an RBN with two states $q_I,q_f$, a cut-off always exists. Whether the cut-off is positive or negative can be decided in \textsf{EXPSPACE}.
\end{restatable}

\section{A Third Model}\label{sec:IO}
We have shown that RBN and ASMS are cube-reachability equivalent and using this we have transferred some results between these two models.
In this section, we will introduce a third model called \emph{Immediate Observation} (IO) nets and show that cube-reachability for IO nets can be reduced to cube-reachability for RBN.
Further, we show that a stronger notion of reduction -- which is satisfied by all the reductions given in this paper -- cannot exist from RBN to IO nets.

\subsection{Immediate Observation Nets}
Immediate observation nets, or \emph{IO nets}, were introduced in~\cite{conf/apn/EsparzaRW19}.
They are a subclass of Petri nets with applications in population protocols and chemical reaction networks. 
An IO net is a Petri net with transitions of a certain shape:
Informally, 
a process (or token) in a state (or place) $p$ \emph{observes} the presence of a process in $q$ and moves to state $p'$, for some states $p,q,p'$ not necessarily distinct.
Because of this, IO nets
can be described in a simpler manner that does not use the full Petri net formalism. 
We will present them this way here, 
to highlight the similarity to the other models and 
to simplify notation.

\begin{definition}
\label{def:io}
An \emph{immediate observation net} is a tuple $\net = (Q,\delta)$ 
where $Q$ is a finite set of states 
and $\delta \subseteq Q \times Q \times Q$ is the transition relation.
\end{definition}

If $(p,q,p') \in \delta$, then we sometimes denote it by $p \act{q} p'$.
A \emph{configuration} $C$ of an IO net $\net$ is a multiset over $Q$.
It intuitively counts the number of processes in each state. 
There is a \emph{step} from a configuration $C$ to a configuration $C'$
if there exists $t=p \trans{q} p' \in \delta$,
such that $C \ge \multiset{p,q}$
and $C' = C - \vec{p} + \vec{p'}$. 
We denote by $C \trans{t} C'$ such a step, and by $\trans{*}$ the reflexive transitive closure of the step relation. 
We can then define runs of an IO net in the usual way.

\subsection{RBN Simulate IO Nets}
\textbf{Construction} \ Let $\net = (Q,\delta)$ be an IO net. 
We construct an RBN that simulates $\net$
in which processes send messages signaling their current state.
Let $\RBN = (Q', \Sigma',\delta')$ be the following RBN: 
The set of states $Q'$ and the alphabet $\Sigma'$ are both  equal to $Q$.
The transition relation $\delta'$ is such that 
for every $q \in Q$ 
there is a transition $q \trans{!q} q$ in $\delta'$, and
for every $p \trans{q} p' \in \delta$ 
there is a transition $p \trans{?q} p'$ in $\delta'$.


\noindent \textbf{Correctness of construction} \ There is a natural bijection between configurations of $\RBN$ and configurations of $\ASMS$. 
If $C$ is a configuration of $\net$, we will abuse notation and denote the corresponding  configuration of $\RBN$ also as $C$.
We now show that $C' \in \poststar_{\net}(C)$ iff $C' \in \poststar_{\RBN} (C)$ for any configurations $C$ and $C'$ of $\net$.
Indeed, if $C$ reaches $C'$ by one step $p\trans{q} p'$ in $\net$, 
then $C \trans{t+t_1} C'$ with $t=q \trans{!q} q$ and $t_1= p \trans{?q} p'$ in $\RBN$. 
Conversely,  let $C \trans{t+t_1,\ldots, t_k} C'$ be a step  in $\RBN$ with
$t= q \trans{!q} q$ and $t_i=p_i \trans{?q} p'_i$ 
for some $k \ge 0$.
The step must be of this form because the only broadcast transitions of $\RBN$
are of the form $q \trans{!q} q$.
Then $C$ reaches $C'$ by the sequence of transitions
$ (p_1 \trans{q} p'_1 ), (p_2 \trans{q} p'_2), \ldots ,( p_k \trans{q} p'_k)$ in $\net$. 

\noindent \textbf{The reduction}  \ With this construction, RBN can simulate IO nets as follows: 
Let $\net$ be an IO net and let $\cube_1 = (L_1,U_1),\cube_1' = (L_1',U_1')$ be two cubes of $\net$. 
Construct the RBN $\RBN$ as described above,
and let $\cube_2 = \cube_1$ and $\cube_2' = \cube_1'$.
By our construction, $\cube_1$ can reach $\cube_1'$ iff $\cube_2$ can reach $\cube_2'$.

\paragraph*{Consequences.}
In~\cite{FSTTCS12}, 
two further restrictions of the unbounded initial cube reachability problem
(presented in Section \ref{subsec:unbounded-cube})
are considered. 
The first restriction, dubbed $CRP[\ge1]$
(where CRP stands for cardinality reachability problem), 
considers only unbounded initial cube pairs $\cube, \cube'$
in which $\cube'=(L',U')$ is such that 
$L'(q) \in \set{0,1}$ and $U'(q) =\infty$ for all $q$.
The second restriction, dubbed $CRP[\ge1,=0]$, 
considers only unbounded initial cube pairs $\cube, \cube'$
in which $\cube'=(L',U')$ is such that 
$L'(q) \in \set{0,1}$ and $U'(q) \in \set{0,\infty}$ for all $q$.
For RBN, the problems $CRP[\ge 1]$ and $CRP[\ge 1,=0]$ are 
shown to be in PTIME and in NP (Theorem 3.3 and 4.3 of~\cite{FSTTCS12}), respectively.
By the construction given above, it is then immediately clear that

\begin{theorem}
\label{thm:io-p}
For IO nets, $CRP[\ge 1]$ and $CRP[\ge 1,=0]$ are in PTIME and in NP respectively.
\end{theorem}

\paragraph*{Strong simulation.}
Consider the following alternative definition of simulation between models $A$ and $B$ (where a model is to be understood as either an RBN, ASMS or IO net): Given an instance $I$ of model $A$ with states $Q_I$,
there exists an instance $J$ of model $B$ with states $Q_J$ such that $J$ is polynomial in the
size of $I$
with $Q_I  \subseteq Q_J$, and
%
there exists a multiset $h$ over $Q_J\setminus Q_I$ of polynomial size
such that 
$C' \in \poststar(C)$ if and only if $C'\cdot h \in \poststar(C\cdot h)$
for any configurations $C,C'$ of $I$.
Notice that strong simulation is a transitive relation. 
%
%
The simulation constructions 
of this paper
verify this strong definition of simulation.

\begin{restatable}{theorem}{StrongSimAll}
\label{thm:strong-sim}
RBN and ASMS strongly simulate each other. Further, IO nets are strongly simulated by RBN (and hence by ASMS as well).
\end{restatable} 
 
We show that this is not the case for IO nets: they cannot strongly simulate RBN (nor ASMS).

\subsection{IO Does not Strongly Simulate RBN}
Assuming that IO nets can strongly simulate RBN,
we will derive a contradiction.
Under this assumption, 
we will first transfer results on the closure of cubes from IO nets to RBN, 
then exhibit a particular RBN which contradicts these results.
We start by recalling definitions and properties relating to cubes.

\paragraph*{Counting sets and norms.}
We consider cubes over a finite set $Q$.
A
finite union of cubes $\bigcup_{i=1}^m (L_i,U_i)$ is called a \emph{counting constraint}
and the set of configurations $\bigcup_{i=1}^m \cube_i$ it describes is called a \emph{counting set}.
We write $\sem{\cC}$ for the counting set described by the counting constraint $\cC$.
Notice that two different counting constraints may describe the same counting set.
For example, let $Q=\set{q}$ and let $(L,U)=(1,3)$, $(L',U')=(2,4)$, $(L'',U'')=(1,4)$. 
The counting constraints $(L,U)\cup(L',U')$ and $(L'',U'')$ define the same counting set.
It is easy to show (see also Proposition 2 of \cite{EsparzaGMW18})
that counting constraints and counting sets are closed under Boolean operations.

Let $\cube=(L,U)$ be a cube.
Let $\lnorm{\cube}$ be the the sum of the components of $L$.
Let $\unorm{\cube}$ be the sum of the finite components of $U$ if there are any, and $0$ otherwise.
We call \emph{norm} of $\cube$ the maximum of $\lnorm{\cube}$ and $\unorm{\cube}$, denoted by $\norm{\cube}$.
We define the norm of a counting constraint $\cC= \bigcup_{i=1}^m \cube_i$ as
$\norm{\cC} \defeq \displaystyle \max_{i\in [1,m]} \{ \norm{\cube_i} \}$.
The norm of a counting set $\cSet$ is the smallest norm of a counting constraint representing $\cSet$, that is, 
$\norm{\cSet} \defeq \displaystyle \min_{\cSet = \sem{\cC}} \{ \norm{\cC} \}$.
Proposition 5 of~\cite{EsparzaGMW18} entails the following results for the norms of the union, intersection and complement.
\begin{prop}%
\label{prop:oponconf}
Let $\cSet_1, \cSet_2$ be counting sets.
The norms of the union, intersection and complement satisfy:
$\norm{\cSet_1 \cup \cSet_2} \leq \max \{\norm{\cSet_1}, \norm{\cSet_2} \}$,
$\norm{\cSet_1 \cap \cSet_2} \leq \norm{\cSet_1} + \norm{\cSet_2}$
and
$\norm{\N^n \setminus \cSet_1} \leq \norm{\cSet_1} + \norm{\cSet_2}$.
%
%
%
\end{prop}

The following result for IO nets is deduced directly from Theorem 6 in \cite{EsparzaRW19}.
It states that the forward and backward reachability set of a counting set is still a counting set, and bounds its norms polynomially.
This result, transferred to RBN under the assumption of a strong simulation,
will amount to a contradiction.

\begin{theorem}
\label{thm:ccreach-io}
Let $\net=(Q,\delta)$ be an IO net, and let $\cSet$ be a counting set of $\net$.
Then $\poststar(\cSet)$ is also a counting set and
$
\norm{\poststar(\cSet)} \leq  \norm{\cSet} + |Q|^3.
$
The same holds for $\prestar(\cSet)$.
\end{theorem}

Assuming that IO nets strongly simulate RBN, we can transfer the result of 
Theorem \ref{thm:ccreach-io} to RBN.
 
\begin{restatable}{theorem}{ThmCCReachRBN}
\label{thm:ccreach-rbn}
Assume that IO nets can strongly simulate RBN.
There exists a constant $k$ such that 
for any RBN $\RBN=(Q,\Sigma,\delta)$, 
for any counting set $\cSet$ of $\RBN$,
 $\poststar(\cSet)$ is also a counting set and 
$
\norm{\poststar(\cSet)}  \in O( \norm{\cSet}  + |Q|)^{k}.
$
The same holds for $\prestar(\cSet)$.
\end{restatable}
\begin{proof}[Proof Sketch.]
\newcommand{\image}{\mathcal{G}}
\newcommand{\complicated}{\mathcal{M}}

It suffices to show the result for $\cSet$ a cube, 
since for a counting set $\cup_{i=1}^n \cube_i$, 
we have $\poststar(\cup_i \cube_i)=\cup_i  \poststar(\cube_i)$.
Fix an RBN $\RBN=(Q,\Sigma,\delta)$ and a cube $\cube$ over $Q$.
Let $\net=(Q_\net, \delta_\net)$ be the IO net of the strong simulation whose existence we assume. 
The definition of strong simulation 
entails the existence of a bijection $b$ from configurations of $\RBN$ to 
a subset $\image$ of ``good" configurations of  $\net$.
The bijection verifies that a cube of $\RBN$ is
mapped to a cube of $\net$, and 
that a cube of $\net$ 
restricted to configurations of $\image$
is mapped to a cube of $\RBN$.

Since $b$ preserves cubes, $b(\cube)$ is a cube.
By Theorem \ref{thm:ccreach-io}, $\poststar(b(\cube))$ is a counting set,
and thus there exist cubes $\cube_1, \ldots, \cube_n$ of $\net$ 
such that $\poststar(b(\cube))=\cup_{i=1}^n \cube_i$.
Let $\complicated$ be the set $\cup_{i=1}^n b^{-1}(\cube_i|_\image)$ of $\RBN$.
We show that $\poststar(\cube)=\complicated$.
Since the $b^{-1}(\cube_i|_\image)$ are cubes by strong simulation,
$\poststar(\cube)$ is a counting set as a union of cubes.
The size of $\poststar(b(\cube))$ is polynomial in $\cube$ 
and $\RBN$ by Theorem  \ref{thm:ccreach-io},
and thus the size of $\poststar(\cube)$ is too.
\end{proof}


\paragraph*{Deriving the contradiction.}
We now exhibit a contradiction to the result of Theorem 
\ref{thm:ccreach-rbn}, thus proving that IO nets do not strongly simulate RBN.
Recall the RBN represented in Figure \ref{fig:rbn}.
We can generalize it to a family of RBN $\RBN_n=(Q,\Sigma,\delta)$, 
parameterized by $n\ge 1$,
with set of states $\set{tok, sent} \cup \set{a_i,b_i,c_i | 1 \le i \le n}$.
Let $\cube_0$ be the cube in which there are arbitrarily many agents in $tok$,
exactly one agent in each $a_i$ and $0$ agents in the other states.
Let $\cube_f$ be the cube in which there is a least one agent in $c_n$
and an arbitrary number elsewhere.
We claim that if we start from a configuration of $\cube_0$,
we can only reach $\cube_f$ if we initially have $2^n$ or more agents in $tok$.
Indeed we can show by  induction on $i \in \set{1,\ldots,n}$
that $1$ must be broadcasted $2^i$ times to reach $c_i$,
and thus that  $2^i$ agents are needed in $tok$ 
initially to reach $c_i$. 
By Proposition \ref{prop:oponconf} and Theorem \ref{thm:ccreach-rbn},
the set $S := \poststar(\cube_0) \cap \cube_f$ is a counting set of size at most polynomial in $|Q|, \norm{\cube_0}$ and  $\norm{\cube_f}$.
The cubes $\norm{\cube_0}$ and  $\norm{\cube_f}$ have norms 
$n$ and $1$  respectively, so  $S$ is of norm polynomial in $n$.
Thus if it is non-empty 
it must contain a configuration of size at most polynomial in $n$:
simply take the configuration equal to the lower bounds $L$ of one of the cubes whose union is the counting set $\poststar(\cube_0) \cap \cube_f$.
This contradicts the fact that $2^n$ agents are needed to reach $\cube_f$.


\paragraph*{Acknowledgements:} 
We would like to thank Javier Esparza and the anonymous reviewers for their useful feedback.

\nocite{*}
\bibliographystyle{eptcs}
\bibliography{refs}

\appendix

\section{Appendix for Section \ref{sec:simulations}}

\LmNormalASMS*
\begin{proof}
    Let $n$ be  the length of $\rho$.
    We proceed by induction on $n$. If $n = 0$, we are done.
    Let $n > 0$ and $\rho = \rho_1,\dots,\rho_n$. Assume now that any run of length strictly less than $n$ can be put in normal form. Notice that $\rho$ begins
    at $(\widehat{C},\#)$ and there is no transition which reads the value $\#$ and all transitions
    which write $\#$ can only do so from an intermediary state. Hence, it must be the 
    case that $\rho_1 := q \act{W(a)} [q,a,q']$ for some $a \in \Sigma$. 
    We now consider two cases:
    
    \textbf{Case 1: } Suppose there is no $i > 1$ such that $\rho_i$ is a transition which
    writes a value $b \neq \#$.  Hence, every transition in $\rho_2,\dots,\rho_n$ either reads the value $a$ or writes $\#$, and since $a$ cannot be read after $\#$ has been written, 
    there must be an index $2 \le j \le n$ such that every transition in $\rho_2,\dots,\rho_{j-1}$
    reads $a$ and every transition in $\rho_j, \dots, \rho_n$ writes $\#$. Now, notice that $\rho$
    begins and ends at good configurations and writing and reading the value $a$ makes a process move 
    into an intermediary state and writing $\#$ makes a process move out of an intermediary 
    state. With these three points, it can then be easily seen that 
    $(\widehat{C},\#) \act{\rho} (\widehat{C'},\#)$ must be a pseudo-step.
    
    \textbf{Case 2: } Suppose there is $i > 1$ such that $\rho_i$ is a transition which writes
    a value $b \neq \#$. By the same argument as before, it is easy to see that 
    there must exist $2 \le j \le i-1$ such that every transition in $\rho_2,\dots,\rho_{j-1}$
    reads $a$ and every transition in $\rho_j, \dots, \rho_{i-1}$ writes $\#$. Let $Z$
    be the configuration reached after $\rho_{i-1}$.
    Let $M = Z(I)$, i.e., $M$ is the multiset of processes at the configuration $Z$
    which are in some intermediary state. Since the only way out of intermediary states is to
    write $\#$ onto the register, if $M = \multiset{[p_1,a,p_1'],\dots,[p_k,a,p_k']}$ then
    there must exist $i < i_1,\dots,i_k$ such that each $\rho_{i_l}$ is $[p_l,a,p_l'] \act{W(\#)} p_l'$. Hence, we can rearrange the run as follows: Let $\rho' := \rho_{i},\dots,\rho_{i_1-1},\rho_{i_1+1},\dots,\rho_{i_2-1},\rho_{i_2+1},\dots,\rho_{i_k-1},\rho_{i_k+1},\dots,\rho_n$. Then it is easy to see that the following is a valid run in $\ASMS$:
    $(\widehat{C},\#) \act{\rho_1,\dots,\rho_{i-1}} Z \act{\rho_{i_1},\rho_{i_2},\dots,\rho_{i_k}} Z' \act{\rho'} (\widehat{C'},\#)$. Since the run between $(\widehat{C},\#)$ and $Z'$ is a pseudo-step,
    by applying induction hypothesis on the rest of the run, we are done.
\end{proof}


\LmNormalRBN*
\begin{proof}
	Suppose $Z_0 := Z \trans{b^1 + r_1^1,\dots,r_{n_1}^1} Z_1 \trans{b^2 + r_1^2,\dots,r_{n_2}^2} Z_2 \dots Z_{m-1} \trans{b^m + r_1^m,\dots,r_{n_m}^m} Z_m := Z'$.
	We proceed by induction on $m$. The case of $m = 0$ is trivial.
	
	Suppose $m > 0$ and assume that the claim is true for all numbers less than $m$. 
	Since $Z_0$ is a good configuration, there are only two possible choices for $b^1$.
	
	\textbf{Case 1: } Suppose $b^1 = a \trans{!Read_a} a$ for some $a \in \Sigma$.
	If $n_1 = 0$, then $Z_1 = Z_0$ and by induction hypothesis, we are done.
	Suppose $n_1 > 0$. For each $1 \le i \le n_1$, let $r_i^1 = (p_i,?Read_a,q_i)$.
	Replace the transition between $Z_0$ and $Z_1$ with,
	$Z_0 \trans{b^1 + r_1^1} C_1 \trans{b^1 + r_2^1} C_2 \trans{b^1 + r_3^1} \dots 
	C_{n_1-1} \trans{b^1 + r_{n_1}^1} Z_1$. Now, $Z_1$ is a good configuration,
	and applying the induction hypothesis on the run between $Z_1$ and $Z'$, we are done.

	\textbf{Case 2: } Suppose $b^1 = d \trans{!Ch_a} \overline{a}$ for some $d,a \in \Sigma$.
	Hence, $Z_1$ is not a good configuration and so $Z_1 \neq Z'$. If $n_1 = 0$, then no process in the configuration $Z_1$ can broadcast any message,
	and so no step from $Z_1$ is possible which leads to a contradiction. Hence, we can assume that $n_1 > 0$.
	
	For each $1 \le i \le n_1$, let $r_i^1 = (p_i,?Ch_a,[p_i,a,q_i])$.
	Notice that the only processes which can broadcast a message at the configuration $Z_1$
	are the processes in the multiset $\multiset{[p_1,a,q_1],\dots,[p_{n_1},a,q_{n_1}]}$.
	Hence $b^2 = (p_i[a]q_i,!Ack_a,q_i)$ for some $i$. Without loss of generality, we can assume that $i = 1$.
	
	Notice that the only process which can receive the message $Ack_a$ at the configuration $Z_1$
	is the process at the state $\overline{a}$.
	Depending on whether this process receives the message at $Z_1$, we either have $n_2 = 0$ or $n_2 = 1$. 
	Hence, we get two subcases:
	
	\textbf{Case 2a): } Suppose $n_2 = 0$. Let $S := \sum_{i=2}^{n_1} \vec{p_i} - \sum_{i=2}^{n_1} \vec{[p_i,a,q_i]}$.
	Then reorder the run between $Z_0$ and $Z_2$ as follows:
	$Z_0 \trans{b^1 + r_1^1} Z_1 - S \trans{b^2 + (\overline{a},?Ack_a,a)} Z_2 - S - \overline{\vec{a}} + \vec{a} 
	\trans{(a,!Ch_a,\overline{a}) + r_2^1,\dots,r_{n_1}^1} Z_2$. Notice that the 
	configuration $Z_2 - S - \overline{\vec{a}} + \vec{a}$ is a good configuration and 
	has a run of length $m-1$ to $Z'$. Applying induction hypothesis, we are then done.
	
	\textbf{Case 2b): } Suppose $n_2 = 1$. Hence, $r_1^2 = (\overline{a},?Ack_a,a)$ and so
	$Z_2(\overline{a}) = 0$ and $Z_2(a) = 1$.
	We consider two further subcases:
	\begin{itemize}
		\item Suppose there exists $\alpha > 2$ such that $Z_{\alpha}(\overline{a}) = 1$.
		Let $\alpha$ be the minimum such index. Hence, there must exist some $d' \in \Sigma$ such that $b^{\alpha}$ is $(d',!Ch_a,\overline{a})$.
		Consider the run $Z_2 \trans{b^3 + r_1^3,\dots,r_{n_3}^3} Z_3 \trans{b^4 + r_1^4,\dots,r_{n_4}^4} \dots Z_{\alpha-1} \trans{b^{\alpha} + r_1^\alpha, \dots,r_{n_\alpha}^\alpha} Z_\alpha$.
		
		Suppose no $b \in \{b^3,\dots,b^{\alpha-1}\}$ is labelled by $!Ack_a$. Observe that $Z_2 \ge \sum_{i=2}^{n_1} \vec{[p_i,a,q_i]}$ and 
		the only way to move a process out of $[p_i,a,q_i]$ is to broadcast the message $!Ack_a$. Hence none of the processes
		in the states $[p_i,a,q_i]$ are involved in any of the transitions between $Z_2$ and $Z_{\alpha}$ and so $Z_j \ge \sum_{i=2}^{n_1} \vec{[p_i,a,q_i]}$ for every $2 \le j \le \alpha$.
		So, we can rearrange the run between $Z_0$ and $Z_\alpha$ as follows: Let $S := \sum_{i=2}^{n_1} \vec{p_i} - \sum_{i=2}^{n_1} \vec{[p_i,a,q_i]} $.
		We have $Z_0 \trans{b^1 + r_1^1} Z_1 - S \trans{b^2 + r_1^2} Z_2 - S \trans{b^3 + r_1^3,\dots,r_{n_3}^3} Z_3 - S \trans{b^4 + r_1^4,\dots,r_{n_4}^4}
		\dots Z_{\alpha-1} - S\trans{b^\alpha, r_1^\alpha, \dots, r_{n_\alpha}^\alpha, r_2^1, \dots, r_{n_1}^1} Z_\alpha$. 
		The configuration $Z_2 - S$ is a good configuration and has a run to $Z'$ of length $m-2$. Applying induction hypothesis, we are then done.
		
		Suppose some $b \in \{b^3,\dots,b^{\alpha-1}\}$ is labelled by $!Ack_a$. Let $b = b^i$ be the first such transition.
		Notice that the only processes in the configurations $Z_2,\dots,Z_{\alpha-1}$ capable of broadcasting the message $Ack_a$
		are the processes in the multiset $\multiset{[p_2,a,q_2],\dots,[p_{n_1},a,q_{n_1}]}$. Hence, $b^i = ([p_j,a,q_j],!Ack_a,q_j)$ for some $j$.
		Without loss of generality let $b^i = ([p_2,a,q_2],!Ack_a,q_2)$. Notice that no process in the configurations $Z_2,\dots,Z_{\alpha-1}$
		receives the message $Ack_a$. (The only state capable of receiving that message is at $\overline{a}$, but by our assumption the first time
		a process at state $\overline{a}$ receives the message $Ack_a$ is at index $\alpha > i$.)
		Hence $n_i = 0$. Therefore $Z_i \trans{([p_2,a,q_2],!Ack_a,q_2)} Z_{i+1}$. With this in mind,
		we can rearrange the run between $Z_0$ and $Z_i$ as follows: Let $S := \sum_{i=2}^{n_1} \vec{p_i} - \sum_{i=2}^{n_1} \vec{[p_i,a,q_i]}$.
		We have $Z_0 \trans{b^1 + r_1^1} Z_1 - S \trans{b^2 + r_1^2} Z_2 - S
		\trans{(a,!Ch_a,\overline{a}) + r_2^1,\dots,r_{n_1}^1} Z_2 - \vec{a} + \overline{\vec{a}} 
		\trans{([p_2,a,q_2],!Ack_a,q_2) + (\overline{a},?Ack_a,a)} Z_2 - \vec{[p_2,a,q_2]} + \vec{p_2}
		\trans{b^3 + r_1^3,\dots,r_{n_3}^3} Z_3 -  \vec{[p_2,a,q_2]} + \vec{p_2} \trans{b^4 + r_1^4,\dots,r_{n_4}^4} 
		\dots Z_{i-2} - \vec{[p_2,a,q_2]} + \vec{p_2} \trans{b^{i-1} + r_1^{i-1} \dots r_{n_{i-1}}^{i-1}} Z_i$.
		Notice that $Z_2 - S$ is a good configuration and has a run to $Z'$ of length $<m$. Applying the induction hypothesis to this run, we are done.
		
		\item Suppose there does not exist $\alpha > 2$ such that $Z_{\alpha}(\overline{a}) = 1$.
		Consider the run $Z_2 \trans{b^3 + r_1^3,\dots,r_{n_3}^3} Z_3 \trans{b^4 + r_1^4,\dots,r_{n_4}^4} \dots Z_{m-1} \trans{b^{m} + r_1^m, \dots,r_{n_m}^m} Z_m$.
		
		Suppose no $b \in \{b^3,\dots,b^{m}\}$ is labelled by $!Ack_a$. By the same argument as the previous case,
		we can conclude that for every $2 \le j \le m$, we have $Z_j \ge \sum_{i=2}^{n_1} \vec{[p_i,a,q_i]}$.
		However $Z_m$ is a good configuration. Hence, it must be the case that $n_1 = 1$. If $n_1 = 1$, notice that
		the configuration $Z_2$ is a good configuration and the run between $Z_0$ and $Z_2$ is already in normal form.
		Because of the induction hypothesis, we are done.
		 
		Suppose some $b \in \{b^3,\dots,b^{m}\}$ is labelled by $!Ack_a$. Let $b = b^i$ be the first such transition.
		Similar to the previous case, we can show that the step after $Z_i$ should be of the form:
		$Z_i \trans{([p_2,a,q_2],!Ack_a,q_2)} Z_{i+1}$. Hence, we can rearrange the run between $Z_0$ and $Z_i$ in the same way as
		before and conclude by applying the induction hypothesis.
	\end{itemize}
\end{proof}

\section{Appendix for Section \ref{sec:consequences}}

\ThmLeaderRBN*
\begin{proof}
For the upper bound, notice that given an RBN-leader cube, our reduction in section~\ref{subsec:ASMS-RBN}
produces an ASMS-leader cube. Since the reachability problem for ASMS-leader cubes is in \textsf{NP}~\cite{JACM16}, this immediately gives us the same upper bound for RBN-leader cube reachability.

For the lower bound, we give a reduction from 3-SAT.
Let $\varphi = \bigwedge_{i=1}^m C_i$ be a 3-CNF formula over the variables $x_1,\dots,x_n$ where each $C_i = \ell_i^1 \lor \ell_i^2 \lor \ell_i^3$. Construct an RBN-leader protocol as follows : 
The leader will have $n+m+1$ states $q_0,q_1,\dots,q_n,p_1,\dots,p_m$ and the contributors will have $2n+1$ states $init,y_1,\bar{y_1},y_2,\bar{y_2},\dots,y_n,\bar{y_n}$.
For every $1 \le i \le n$, from the state $q_{i-1}$, the leader can either broadcast the message $\top_i$ or $\bot_i$ and move to $q_{i}$. 
These two transitions intuitively correspond to the leader guessing that the variable $x_i$ is either true or false.
From the state $init$, a contributor can move to $y_i$ if it receives the message $\top_i$ or move to $\bar{y_i}$ if it receives the message $\bot_i$. Intuitively, for each $i$, some contributor
receives the guess made by the leader and stores it in its finite set of states so that
it can send it to the leader later on.

For every $1 \le i \le n$, there is a self-loop at the state $y_i$ (resp. $\bar{y_i}$), which can broadcast the message $x_i$ (resp. $\bar{x_i}$). This corresponds to the contributors broadcasting to the leader
the value guessed by it for the $i^{th}$ variable.
Denoting the state $q_n$ by $p_0$, for every $1 \le i \le m$, from the state $p_{i-1}$, the leader can receive any one of the messages $\ell_i^1, \ell_i^2, \ell_i^3$ and move to $p_i$. Hence, the
leader can move to $p_i$ from $p_{i-1}$ iff the guesses that it made before, satisfy the $i^{th}$ clause.

If we now set initially the leader must start at $q_0$ and end up at $p_m$ and the contributors must start at $init$, then it is clear from construction that the RBN-leader reachability problem for this leader protocol is true iff $\varphi$ is satisfiable. 
\end{proof}

\ThmASCoverabilityRBN*
\begin{proof}
We consider the reduction given in section~\ref{subsec:ASMS-RBN}, which given an RBN $\RBN$ 
constructs an ASMS $\ASMS$ which contains all the states of $\RBN$. We notice the following points regarding that construction.
\begin{itemize}
    \item \emph{Remark 1: } To every configuration $C$ of $\RBN$, our construction uniquely identifies a configuration $\hat{C}$ of $\ASMS$ (the set of ``good'' configurations) such that $\hat{C}(q) = C(q)$ for every state $q$ of $\RBN$ and 
	$C \act{*} C'$ in $\RBN$ iff $\hat{C} \act{*} \hat{C'}$ in $\ASMS$. 
	Moreover, the size of $C$ and $\hat{C}$ are the same.
	\item \emph{Remark 2: } If $\hat{C} \act{*} C'$ is a run in $\ASMS$ for some configuration $C'$ which is not good, then it is possible to extend this run to a good configuration $\hat{C} \act{*} C' \act{*} \hat{C''}$ such that $\hat{C''}(q) \ge C'(q)$ for every state $q$ of $\RBN$. (Indeed, as long as there is a process in some intermediate state $[p,a,p']$, we can use the $[p,a,p'] \act{W(\#)} p'$ transition to move the process into the state $p'$).
\end{itemize}

With these remarks, we now give the proof of the theorem. Let $\RBN$ be an RBN with two states
$q_I$ and $q_f$. For any $k \ge 1$,  let $C^k = \multiset{k \cdot q_I}$ be a configuration of $\RBN$.
Let $\ASMS$ be the ASMS constructed using the reduction from $\RBN$ in section~\ref{subsec:ASMS-RBN}.
By construction, it is easy to see that $\hat{C^k}$ is simply the configuration $(\multiset{k \cdot q_I},\#)$.

We now claim that,
\begin{quote}
	For any $k \ge 1$, $\post^{*}_{\RBN}(C^k) \subseteq \pre^{*}_{\RBN}(\uparrow q_f)$ iff $\post^{*}_{\ASMS}(\hat{C^k}) \subseteq \pre^{*}_{\ASMS}(\uparrow q_f)$
\end{quote}

Indeed, suppose $\post^{*}_{\RBN}(C^k) \subseteq \pre^{*}_{\RBN}(\uparrow q_f)$. To prove that $\post^{*}_{\ASMS}(\hat{C^k}) \subseteq \pre^{*}_{\ASMS}(\uparrow q_f)$ we need
to show that from every $Z \in \post^{*}_\ASMS(\hat{C^k})$, we can reach a configuration $Z' \in \uparrow q_f$.
Let $Z \in \post^{*}_\ASMS(\hat{C^k})$. By remark 2, there exists 
a configuration $C'$ of $\RBN$ such that $Z$ can reach the configuration $\hat{C'}$ in $\ASMS$ and so $\hat{C^k}$ can reach $\hat{C'}$. By remark 1, we 
have that $C^k$ can reach $C'$ in $\RBN$ and so (by our assumption), from $C'$ it must be possible to reach some configuration $C'' \in \uparrow q_f$ in $\RBN$.
By remark 1, we then have that $\hat{C'}$ can reach $\hat{C''}$ in $\ASMS$. Since $\hat{C''}(q_f) = C''(q_f)$ it follows that $\hat{C''} \in \uparrow q_f$.
Hence, from the configuration $Z$ we have managed to reach a configuration in $\uparrow q_f$ in $\ASMS$, which is what we wanted to prove.

Suppose  $\post^{*}_{\ASMS}(\hat{C^k}) \subseteq \pre^{*}_{\ASMS}(\uparrow q_f)$. To prove that $\post^{*}_{\RBN}(C^k) \subseteq \pre^{*}_{\RBN}(\uparrow q_f)$ we need
to show that from every $C \in \post^*_\RBN(C^k)$, we can reach a configuration $C' \in \uparrow q_f$.
Let $C \in \post^*_\RBN(C^k)$. By remark 1, we have that $\hat{C^k}$ can reach $\hat{C}$ and so (by our assumption),
from $\hat{C}$ it must be possible to reach some configuration $Z \in \uparrow q_f$. By remark 2, there exists a configuration $C'$ of $\RBN$ such that $Z$ can reach the configuration $\hat{C'}$ 
and $\hat{C'}(q_f) \ge Z(q_f) \ge 1$. Hence, $\hat{C}$ can reach the configuration $\hat{C'}$ in $\ASMS$ and so by remark 1, $C$ can reach $C'$ in $\RBN$.
Since $\hat{C'}(q_f) = C'(q_f) \ge 1$, it follows that $C' \in \uparrow q_f$. Hence, from the configuration $C$ we have managed
to reach a configuration in $\uparrow q_f$ in $\RBN$, which is what we wanted to prove.

Now, by Theorem~\ref{thm:ICALPPatricia}, we know that every ASMS has either a positive or a negative cut-off and deciding whether a given ASMS has a positive cut-off is in \textsf{EXPSPACE}.
Combining this fact along with the argument given above, we can conclude that every RBN has either a
positive or a negative cut-off and deciding whether a given RBN has a positive cut-off is in \textsf{EXPSPACE}.

\end{proof}

\section{Appendix for Section \ref{sec:IO}}

\StrongSimAll*
\begin{proof}
The simulation construction from IO nets to RBN
constructs an RBN with same state set $Q$ as the IO net,
and maps a configuration $C$ over $Q$ to the same configuration.
The simulation construction from RBN to ASMS
constructs an ASMS with state set $Q' \supseteq Q$ for $Q$ the original state set.
Adapting notation, we can see a configuration $C=(M,d)$
of an ASMS as the multiset $M \cdot d$ over $Q' \cup \Sigma'$.
The construction maps a configuration $C$ over $Q$ to the configuration 
$C \cdot \vec{0}_{Q' \setminus Q} \cdot \#$ over $Q' \cup \Sigma'$,
where $\vec{0}_{Q' \setminus Q}$ is the zero multiset over $Q' \setminus Q$.
Finally, the simulation construction from ASMS to RBN
constructs an RBN with state set $Q' \supseteq (Q \cup \Sigma)$ for $Q$ the original state set.
The construction maps 
a configuration $C=M\cdot d$ over $Q\cup \Sigma$ to the configuration 
$C \cdot \vec{d} \cdot \vec{0}_{Q' \setminus (Q \cup \Sigma)}$ over $Q'$,
where $\vec{d}$ is the multiset over $\Sigma$ equal to $1$ on $d$ and $0$ elsewhere, 
and $\vec{0}_{Q' \setminus (Q \cup \Sigma)}$
is the zero multiset over $Q' \setminus (Q \cup \Sigma)$.
\end{proof}

\ThmCCReachRBN*
\begin{proof}
\newcommand{\image}{\mathcal{G}}
\newcommand{\complicated}{\mathcal{M}}
We start by a remark on the strong simulation definition:
The definition of strong simulation between an instance $I$ of model $A$ and
an instance $J$ of model $B$ 
entails a bijection $b$ from configurations of instance $I$ to 
a subset $\image$ of ``good" configurations of instance $J$.
The bijection verifies that a cube of $I$ is
mapped to a cube of $J$, and 
that a cube of $J$ 
restricted to configurations of $\image$
is mapped to a cube of $I$.
Intuitively, the image of a cube $\cube$ of $I$ 
is its ``concatenation" with the cube on $Q_J\setminus Q_I$ of lower and upper bound equal to $h$.
The norm of $b(\cube)$ is $\norm{b(\cube)}=\norm{\cube}+|h|$.
A cube $\cube$ of $J$ restricted to $\image$
is equal to the cube $\cube \cap \mathcal{H}$ where
$\mathcal{H}$ is the cube 
of lower bound $0$ and upper bound $\infty$ on $Q_I$, and upper and lower bounds equal to $h$ on $Q_J\setminus Q_I$. 
The reverse image of this cube is the cube of $I$ in which we ``forget" the information of $Q_J\setminus Q_I$. 
The norm of $b^{-1}(\cube|_\image)$ is $\norm{b^{-1}(\cube|_\image)} \le \norm{\cube}$.
\chana{L,U for the empty cube not defined yet} 

Now, it suffices to show the claim of the theorem when $\cSet$ is a cube, 
since for a counting set $\cup_{i=1}^n \cube_i$, 
we have $\poststar(\cup_i \cube_i)=\cup_i  \poststar(\cube_i)$.
Fix an RBN $\RBN=(Q,\Sigma,\delta)$ and a cube $\cube$ over $Q$.
Let $\net=(Q_\net, \delta_\net)$ be the IO net of the strong simulation whose existence we assume, and  let $b$ be the bijection induced by the simulation 
from configurations of $\RBN$ to a subset of ``good configurations" of $\net$.
Let us note $\image$ the image by $b$ of the configurations of $\RBN$.

Since $b$ preserves cubes, $b(\cube)$ is a cube.
By Theorem \ref{thm:ccreach-io}, $\poststar(b(\cube))$ is a counting set,
and thus there exist cubes $\cube_1, \ldots, \cube_n$ of $\net$ 
such that $\poststar(b(\cube))=\cup_{i=1}^n \cube_i$.
Let $\complicated$ be the set $\cup_{i=1}^n b^{-1}(\cube_i|_\image)$ of $\RBN$.
We show that $\poststar(\cube)=\complicated$.
Let $C\in \complicated$. 
There exists $i$ such that $C \in  b^{-1}(\cube_i|_\image)$.
Thus $b(C) \in  \cube_i|_\image \subseteq \poststar(b(\cube))$.
By strong simulation, 
$C \in \poststar(\cube)$.
For the other direction  of inclusion, consider 
$C\in  \poststar(\cube)$.
By strong simulation, $b(C) \in \poststar(b(\cube))$,
and there exists $i$ such that $b(C) \in \cube_i$.
By definition $b(C) \in \image$, so $b(C) \in  \cube_i|_\image$ 
and $C \in b^{-1}(\cube_i|_\image) \subseteq \complicated$,
concluding our proof of equality.

Since the $b^{-1}(\cube_i|_\image)$ are cubes by strong simulation,
$\poststar(\cube)$ is a counting set as a union of cubes.
The size of $\poststar(b(\cube))$ is polynomial in $\cube$ and $\RBN$
 by Theorem  \ref{thm:ccreach-io},
and thus the size of $\poststar(\cube)$ is too.
\end{proof}

\end{document}